\documentclass[12pt,reqno]{amsart}

  \usepackage{hyperref}
\usepackage[usenames]{color}
\usepackage[latin1]{inputenc} 
\usepackage{float}

\usepackage[english]{babel} 

\usepackage{amssymb, amsmath}
\usepackage{geometry,graphicx}

\theoremstyle{plain}
\newtheorem{theorem}{Theorem}[section] 

\newtheorem{lemma}[theorem]{Lemma} 

\newtheorem{corollary}[theorem]{Corollary}

\newtheorem{remark}[theorem]{Remark}

\newtheorem{example}[theorem]{Example}
\newtheorem{proposition}[theorem]{Proposition}

\numberwithin{equation}{section} 
\newcommand{\Cov}{\mathbb{C}\mathrm{ov}}
\newcommand{\e}{\mathrm{e}}
\newcommand{\dif}{\mathrm{d}}
\newcommand{\leb}{\mathrm{L}}

\title[\hfill\protect\parbox{0.95\linewidth}{Approximation of the probability density function of the ran- \\domized heat equation with non-homogeneous boundary conditions}]{Approximation of the probability density function of the randomized heat equation with non-homogeneous boundary conditions}

\author{J. Calatayud, J.-C. Cort\'{e}s, M. Jornet} 

\begin{document}

\maketitle

\begin{center}
\noindent
\address{Instituto Universitario de Matem\'{a}tica Multidisciplinar,\\
Universitat Polit\`{e}cnica de Val\`{e}ncia,\\
Camino de Vera s/n, 46022, Valencia, Spain\\
email: jucagre@alumni.uv.es; jccortes@imm.upv.es; marjorsa@doctor.upv.es
}
\end{center}

\begin{abstract}

This paper deals with the randomized heat equation defined on a general bounded interval $[L_1,L_2]$ and with non-homogeneous boundary conditions. The solution is a stochastic process that can be related, via changes of variable, with the solution stochastic process of the random heat equation defined on $[0,1]$ with homogeneous boundary conditions. Results in the extant literature establish conditions under which the probability density function of the solution process to the random heat equation on $[0,1]$ with homogeneous boundary conditions can be approximated. Via the changes of variable and the Random Variable Transformation technique, we set mild conditions under which the probability density function of the solution process to the random heat equation on a general bounded interval $[L_1,L_2]$ and with non-homogeneous boundary conditions can be approximated uniformly or pointwise. Furthermore, we provide sufficient conditions in order that the expectation and the variance of the solution stochastic process can be computed from the proposed approximations of the probability density function. Numerical examples are performed in the case that the initial condition process has a certain Karhunen-Loève expansion, being Gaussian and non-Gaussian. \\
\\
\textit{Keywords:} Stochastic calculus, Random heat equation, Non-homogeneous boundary conditions, Random Variable Transformation technique, Karhunen-Lo\`{e}ve expansion, Probability density function, Numerical simulations. \\
\\
\textit{AMS 2010:} 34F05, 60H35, 65Z05, 60H15, 93E03.
\end{abstract}

\section{Introduction and motivation}\label{sec_introduccion_motivacion}

It is well-known that the heat equation plays a key role to describe mathematically diffusion processes. Due to  heterogeneity and impurities in materials and  errors in the temperature measurements, many authors have proposed to treat the diffusion coefficient, initial and/or boundary conditions in the heat equation  as random variables and/or stochastic processes rather than deterministic constants and functions, respectively. This approach leads to stochastic and random heat equation formulation \cite[pp.~96--97]{llibre_smith}. In the former case, the stochastic heat differential equation is forced by an irregular  stochastic process such as a Brownian White Noise process \cite[Ch.~4]{Oksendal_otros}. These kind of equations are typically written in terms of stochastic differentials and interpreted as It\^{o} or Stratonovich integrals. Special stochastic calculus is usually applied to obtain exact or approximate solutions to this class of differential equations \cite{Oksendal_solo, Oksendal_otros, Kloeden_Platen}.    

In \cite{Xu_APM_2014}, a new stochastic analysis for steady and transient one-dimensional heat conduction problem based on the homogenization approach is proposed. Thermal conductivity is assumed to be a random field depending on a finite number random variables.  Both mean and variance of stochastic solutions are obtained analytically for the field consisting of independent identically distributed random variables.  In \cite{Chiba_APM_2009}, the stochastic temperature field is analyzed by considering the annular disc to be multi-layered with spatially constant material properties and spatially constant but random heat transfer coefficients in each layer. A type of integral transform method together with a perturbation technique are employed in order to obtain the analytical solutions for the statistics (mean and variance) of the temperature. Another fruitful approach to deal with different formulations of random heat equations is the Mean Square Calculus \cite[Ch.~4]{Soong}. In \cite{CCCJ_APM_2014}, an analytic-numerical mean square solution of the random diffusion model in an infinite medium is constructed by applying the random exponential Fourier integral transform. A complementary analysis, based on random trigonometric Fourier integral transforms, to solve random partial differential heat problems with non-homogeneous boundary value conditions has been  presented in \cite{CCJ_JCAM_2016}. In these two latter contributions, reliable approximations for the mean and the variance of the solution stochastic process are provided. Likewise asymptotic-preserving methods for random hyperbolic, transport equations and  radiative heat transfer equations with random inputs and diffusive scalings have been recently studied using generalized polynomial chaos based stochastic Galerkin method \cite{Jin_Xiu_Zhu_JCP_2015, Jin_Lu_JCP_2017}. The probabilistic information to the solution stochastic process of the random heat equation in all the aforementioned contributions focus on first statistical moments like the mean and the variance functions. Nevertheless, a more ambitious target is to determine exact or reliable approximations of the first probability density function of the solution stochastic process, since from it all one-dimensional statistical  moments can be obtained, if they exist. In particular, the mean and the variance can be straightforwardly derived via integration of the first probability density function \cite[Ch.~3]{Soong}. In the context of random partial differential equations, some recent contributions addressing this significant problem include, for example, \cite{Dorini_JCP, Selim_AMC_2011,  Selim_EPJP_2015, Xu_AMM_2016} (see also references therein). From a general standpoint, this paper is aimed to contribute further the study of methods to  determine rigorous approximations of the first probability density function of random partial differential equations focusing on the random heat equation. It is important to point out that the subsequent analysis is  based upon our previous contribution  \cite{jjcm}.
 
In \cite{jjcm}, we have studied the randomized heat equation on the spatial domain $[0,1]$  with homogeneous boundary conditions and assuming that the diffusion coefficient is a positive random variable and that the initial condition is a stochastic process. In a first step, the solution stochastic process of that stochastic problem  was rigorously constructed using two different approaches, namely, the Sample  Calculus and the Mean Square Calculus. The second, and main step, consisted of constructing approximations of the probability density function of the solution by combining the Random Variable Transformation technique and the Karhunen-Lo\`{e}ve expansion. Several results providing sufficient conditions to guarantee the pointwise and uniform convergence of these approximations were established. The aim of this contribution is to extend the study to the case where boundary conditions are random variables and assuming that the problem is stated on an arbitrary interval, say $[L_1,L_2]$. Since this extension depends heavily on some results established in \cite{jjcm}, for the sake of completeness down below we summarize them.       

\section{Preliminaries}

\subsection{Notation} \ \\

For the sake of completeness, this section is devoted to summarize the notation
that will be used throughout this paper.   If $(S,\mathcal{A},\mu)$ is a measure space, we denote by $\leb^p(S)$ or $\leb^p(S,\dif\mu)$ ($1\leq p<\infty$) the set of measurable functions $f:S\rightarrow\mathbb{R}$ such that $\|f\|_{\leb^p(S)}:=(\int_S |f|^p\,\dif \mu)^{1/p}<\infty$. We denote by $\leb^\infty(S)$ or $\leb^\infty(S,\dif\mu)$ the set of measurable functions such that $\|f\|_{\leb^\infty(S)}:=\inf\{\sup\{|f(x)|:\,x\in S\backslash N\}:\,\mu(N)=0\}<\infty$. This norm is usually referred to as ``essential supremum". Hereinafter, we will write a.e. as a short notation for ``almost every'', which means that some property holds except for a set of measure zero. In this paper, we will deal with the following particular cases: $S=\mathcal{T}\subseteq\mathbb{R}$ and $\dif\mu=\dif x$ the Lebesgue measure, $S=\Omega$ and $\mu=\mathbb{P}$ a probability measure, and $S=\mathcal{T}\times\Omega$ and $\dif\mu=\dif x\times \dif\mathbb{P}$. Notice that $f\in \leb^p(\mathcal{T}\times\Omega)$ if and only if $\|f\|_{\leb^p(\mathcal{T}\times \Omega)}=(\mathbb{E}[\int_{\mathcal{T}} |f(x)|^p\,\dif x])^{1/p}<\infty$, where $\mathbb{E}[\cdot]$ denotes the expectation operator. In the particular case of $S=\Omega$ and $\mu=\mathbb{P}$, the short notation a.s. stands for ``almost surely''.

\subsection{Preliminaries on the randomized heat equation on $[0,1]$ with homogeneous boundary conditions} \ \\

Reference \cite{jjcm} provides the necessary results on the approximation of the probability density function of the randomized heat equation on  the spatial domain $[0,1]$ with homogeneous boundary conditions. The main goal of this contribution is to extend these results to the randomized heat equation on a general interval $[L_1,L_2]$ with random boundary conditions.

For the sake of completeness in the presentation, we first recall some deterministic
results that will be useful to achieve our goal. In a deterministic setting, the heat equation on the spatial domain $[0,1]$ with homogeneous boundary conditions has the form
\begin{equation} 
\left\{
\begin{array}{cclcl} 
v_t &=&\beta^2 v_{xx},& & 0<x<1,\;t>0, \\ 
v(0,t)&=&v(1,t)=0,& &t\geq0, \\ v(x,0)&=&\psi(x),& & 0\leq x\leq 1, 
\end{array} 
\right.
\label{edp_determinista_homo}
\end{equation}
where the diffusion coefficient is $\beta^2>0$ and the initial condition is given by the function $\psi(x)$. The classical method of separation of variables provides the formal solution 
\begin{equation}
v(x,t)=\sum_{n=1}^\infty A_n\,\e^{-n^2\pi^2\beta^2 t}\sin(n\pi x),
\label{sol_homo}
\end{equation}
where the Fourier coefficient
\begin{equation}
A_n=2\int_0^1 \psi(y) \sin(n\pi y)\,\dif y 
\label{Fourier_coeff}
\end{equation}
is understood as a Lebesgue integral. As it is proved in \cite[Th. 1.1]{jjcm}, under some simple hypotheses the formal solution (\ref{sol_homo})--(\ref{Fourier_coeff}) becomes a rigorous solution to the deterministic heat equation (\ref{edp_determinista_homo}).

\begin{theorem} \label{c1}
If $\psi$ is continuous on $[0,1]$, piecewise $C^1$ on $[0,1]$ and $\psi(0)=\psi(1)=0$, then (\ref{sol_homo})--(\ref{Fourier_coeff}) is continuous on $[0,1]\times [0,\infty)$, is of class $C^{2,1}$ on $(0,1)\times(0,\infty)$ and is a classical solution of (\ref{edp_determinista_homo}).
\end{theorem}

Now we consider the deterministic PDE problem (\ref{edp_determinista_homo}) in a random setting. Let $(\Omega,\mathcal{F},\mathbb{P})$ be a complete probability space, where $\Omega$ is the sample space, which consists of outcomes that will be denoted by $\omega$, $\mathcal{F}$ is a $\sigma$-algebra of events and $\mathbb{P}$ is the probability measure. The diffusion coefficient is considered as a random variable $\beta^2(\omega)$ and the initial condition is a stochastic process
\[ \psi=\{\psi(x)(\omega):\,x\in [0,1],\,\omega\in\Omega\} \]
in the underlying probability space $(\Omega,\mathcal{F},\mathbb{P})$. The solution (\ref{sol_homo})--(\ref{Fourier_coeff}) becomes a stochastic process in this random scenario, expressed as the formal random series
\begin{equation}
v(x,t)(\omega)=\sum_{n=1}^\infty A_n(\omega)\,\e^{-n^2\pi^2\beta^2(\omega) t}\sin(n\pi x),
\label{sol_homo_random}
\end{equation}
where the random Fourier coefficient
\begin{equation}
 A_n(\omega)=2\int_0^1 \psi(y)(\omega) \sin(n\pi y)\,\dif y 
 \label{An}
\end{equation}
is understood as a Lebesgue integral for each $\omega\in\Omega$ (this is sometimes referred to as SP integral, see \cite[Def.A--1]{strand}). The following result establishes in which sense and under which assumptions the stochastic process (\ref{sol_homo_random})--(\ref{An}) is a rigorous solution to the randomized PDE problem (\ref{edp_determinista_homo}) \cite[Th. 1.3]{jjcm}.

\begin{theorem} \label{st_sol}
The following statements hold:\\
\vspace{-0.5cm}
\begin{itemize}
\item[i)] a.s. solution: Suppose that $\psi\in \leb^2([0,1]\times\Omega)$. Then the random series that defines (\ref{sol_homo_random})--(\ref{An}) converges a.s. for all $x\in [0,1]$ and $t>0$. Moreover, 
\[ v_t(x,t)(\omega)=\beta^2(\omega)\,v_{xx}(x,t)(\omega) \]
a.s. for $x\in (0,1)$ and $t>0$, where the derivatives are understood in the classical sense; $v(0,t)(\omega)=v(1,t)(\omega)=0$ a.s. for $t\geq0$; and $v(x,0)(\omega)=\psi(x)(\omega)$ a.s. for a.e. $x\in [0,1]$.

\item[ii)] $\leb^2$ solution: Suppose that $\psi\in \leb^2([0,1]\times\Omega)$ and $0<a\leq\beta^2(\omega)\leq b$, a.e. $\omega\in\Omega$, for certain $a,b\in\mathbb{R}$. Then the random series that defines (\ref{sol_homo_random})--(\ref{An}) converges in $\leb^2(\Omega)$ for all $x\in [0,1]$ and $t>0$. Moreover, 
\[ v_t(x,t)(\omega)=\beta^2(\omega)\,v_{xx}(x,t)(\omega) \]
a.s. for $x\in (0,1)$ and $t>0$, where the derivatives are understood in the mean square sense (see Subsection \ref{preL}); $v(0,t)(\omega)=v(1,t)(\omega)=0$ a.s. for $t\geq0$; and $v(x,0)(\omega)=\psi(x)(\omega)$ a.s. for a.e. $x\in [0,1]$.
\end{itemize}
\end{theorem}

The main goal of \cite{jjcm} consists of approximating the probability density function of the stochastic process $v(x,t)(\omega)$ given in (\ref{sol_homo_random})--(\ref{An}), for $0<x<1$ and $t>0$. For that purpose, the truncation
\begin{equation}
 v_N(x,t)(\omega)=\sum_{n=1}^N A_n(\omega)\,\e^{-n^2\pi^2\beta^2(\omega) t}\sin(n\pi x) 
 \label{vN}
\end{equation}
is used. Using the Random Variable Transformation technique, see Lemma \ref{lema_abscont}, the density of the truncation $v_N(x,t)(\omega)$ is computed and one proves the following result \cite[Th.~2.8]{jjcm}, which provides conditions under which the density function of the solution stochastic process $v(x,t)(\omega)$ from (\ref{sol_homo_random})--(\ref{An}) can be approximated. The notation $f_X$ stands for the probability density function of the random variable $X$.

\begin{theorem} \label{teor2}
Let $\{\psi(x):\,0\leq x\leq 1\}$ be a process in $\leb^2([0,1]\times \Omega)$. Suppose that $\beta^2$, $A_1$ and $(A_2,\ldots,A_N)$ are independent and absolutely continuous, for $N\geq2$. Suppose that the probability density function $f_{A_1}$ is Lipschitz on $\mathbb{R}$. Assume that $\sum_{n=m}^\infty \|\e^{-(n^2-2)\pi^2\beta^2 t}\|_{\leb^1(\Omega)}<\infty$, for certain $m\in\mathbb{N}$. Then the density of $v_N(x,t)(\omega)$,
\small
\begin{align}
 f_{v_N(x,t)}(v)= {} & \int_{\mathbb{R}^N} f_{A_1} \left(\frac{\e^{\pi^2\beta^2t}}{\sin(\pi x)}\left\{v-\sum_{n=2}^N a_n \e^{-n^2\pi^2\beta^2 t}\sin(n \pi x)\right\}\right)f_{(A_2,\ldots,A_N)} (a_2,\ldots,a_N) \nonumber \\
\cdot & f_{\beta^2}(\beta^2)\frac{\e^{\pi^2\beta^2 t}}{\sin(\pi x)}\,\dif a_2\cdots \dif a_N\,\dif \beta^2, \label{fr}
\end{align}
\normalsize 
converges in $\leb^\infty(\mathbb{R})$ to a density of the random variable $v(x,t)(\omega)$ given in (\ref{sol_homo_random})--(\ref{An}), for $0<x<1$ and $t>0$.
\end{theorem}

Moreover, from the proofs in \cite[Th.~2.7, Th.~2.8]{jjcm}, one has the following rate of convergence for $\{f_{v_N(x,t)}(v)\}_{N=1}^\infty$ under the assumptions of Theorem \ref{teor2}:
\begin{equation}
 |f_{v_N(x,t)}(v)-f_{v(x,t)}(v)|\leq \frac{2\|\psi\|_{\leb^2([0,1]\times\Omega)}L}{\sin^2(\pi x)} \sum_{n=N+1}^\infty \|\e^{-(n^2-2)\pi^2\beta^2 t}\|_{\leb^1(\Omega)},
 \label{rate}
\end{equation}
where $L$ is the Lipschitz constant of $f_{A_1}$.

Another result that could have been added to \cite{jjcm} is presented in what follows. It substitutes the Lipschitz hypothesis by the weaker assumption of a.e. continuity and essential boundedness. The hypothesis $\sum_{n=m}^\infty \|\e^{-(n^2-2)\pi^2\beta^2 t}\|_{\leb^1(\Omega)}<\infty$ is substituted by $\mathbb{E}[\e^{\pi^2\beta^2t}]<\infty$. Then one proves pointwise convergence of the sequence (\ref{fr}), so the uniform convergence on $\mathbb{R}$ and the rate of convergence (\ref{rate}) are lost.

\begin{remark} \label{xyabs} \normalfont
Let $X$ and $Y$ be two independent random variables. If $X$ is absolutely continuous, then $X+Y$ is absolutely continuous. Indeed, for any Borel set $A$, by the convolution formula \cite[p.266]{bil} we have $\mathbb{P}(X+Y\in A)=\int_{\mathbb{R}} \mathbb{P}(X\in A-y)\mathbb{P}_Y(\dif y)$, where $\mathbb{P}_Y=\mathbb{P}\circ Y^{-1}$ is the law of $Y$. If $A$ is null, then $A-y$ is null, so $\mathbb{P}(X\in A-y)=0$. Thus, if $A$ is null, then $\mathbb{P}(X+Y\in A)=0$. By the Radon-Nikodym Theorem \cite[Ch.14]{Radon}, $X+Y$ has a density.
\end{remark}

\begin{theorem} \label{teorllei}
Let $\{\psi(x):\,0\leq x\leq 1\}$ be a process in $\leb^2([0,1]\times \Omega)$. Suppose that $\beta^2$, $A_1$ and $(A_2,\ldots,A_N)$ are independent and absolutely continuous, for $N\geq2$. Suppose that the probability density function $f_{A_1}$ is a.e. continuous on $\mathbb{R}$ and $\|f_{A_1}\|_{\leb^\infty(\mathbb{R})}<\infty$. Assume that $\mathbb{E}[\e^{\pi^2\beta^2t}]<\infty$. Then the density of $v_N(x,t)(\omega)$ given by (\ref{fr}),
\small
\begin{align*}
 f_{v_N(x,t)}(v)= {} & \int_{\mathbb{R}^N} f_{A_1} \left(\frac{\e^{\pi^2\beta^2t}}{\sin(\pi x)}\left\{v-\sum_{n=2}^N a_n \e^{-n^2\pi^2\beta^2 t}\sin(n \pi x)\right\}\right)f_{(A_2,\ldots,A_N)} (a_2,\ldots,a_N)  \\
\cdot & f_{\beta^2}(\beta^2)\frac{\e^{\pi^2\beta^2 t}}{\sin(\pi x)}\,\dif a_2\cdots \dif a_N\,\dif \beta^2,
\end{align*}
\normalsize 
converges pointwise to a density of the random variable $v(x,t)(\omega)$ given in (\ref{sol_homo_random})--(\ref{An}), for all $0<x<1$ and $t>0$.
\end{theorem}
\begin{proof}
Fix $0<x<1$, $t>0$ and $v\in\mathbb{R}$. From (\ref{fr}), notice that
\[ f_{v_N(x,t)}(v)=\mathbb{E}\left[f_{A_1} \left(\frac{\e^{\pi^2\beta^2t}}{\sin(\pi x)}\left\{v-\sum_{n=2}^N A_n \e^{-n^2\pi^2\beta^2 t}\sin(n \pi x)\right\}\right)\frac{\e^{\pi^2\beta^2 t}}{\sin(\pi x)}\right]. \]
Define the random variables
\[ Z_N(\omega):=\frac{\e^{\pi^2\beta^2(\omega)t}}{\sin(\pi x)}\left\{v-\sum_{n=2}^N A_n(\omega) \e^{-n^2\pi^2\beta^2(\omega) t}\sin(n \pi x)\right\},\quad Y(\omega):=\frac{\e^{\pi^2\beta^2(\omega) t}}{\sin(\pi x)}. \]
By Theorem \ref{st_sol} i), we know that 
\[ \lim_{N\rightarrow\infty} Z_N(\omega)=\frac{\e^{\pi^2\beta^2(\omega)t}}{\sin(\pi x)}\left\{v-\sum_{n=2}^\infty A_n(\omega) \e^{-n^2\pi^2\beta^2(\omega) t}\sin(n \pi x)\right\}=:Z(\omega), \]
for a.e. $\omega\in\Omega$. 

By Remark \ref{xyabs}, $Z$ is absolutely continuous. Thus, since $f$ is a.e. continuous, the probability that $Z$ belongs to the discontinuity set of $f$ is $0$. By the Continuous Mapping Theorem \cite[p.7, Th. 2.3]{vaart},
\[ \lim_{N\rightarrow\infty} f_{A_1}(Z_N(\omega))Y(\omega)=f_{A_1}(Z(\omega))Y(\omega), \]
for a.e. $\omega\in\Omega$. Moreover, $|f_{A_1}(Z_N(\omega))Y(\omega)|\leq \|f_{A_1}\|_{\leb^\infty(\Omega)} |Y(\omega)|$, being $Y\in\leb^1(\Omega)$ by the assumption $\mathbb{E}[\e^{\pi^2\beta^2t}]<\infty$. By the Dominated Convergence Theorem \cite[result 11.32, p.321]{rudin}, 
\[ \lim_{N\rightarrow\infty} f_{v_N(x,t)}(v)=\mathbb{E}[f_{A_1}(Z)Y]=:g_{x,t}(v). \]

To conclude, we need to show that $g_{x,t}$ is a density of the random variable $v(x,t)(\omega)$ given by (\ref{sol_homo_random})--(\ref{An}). This is done in a similar way to the last part of the proof of Theorem~2.4 in \cite{jjcm}. We know that, for each $0<x<1$ and $t>0$, $v_N(x,t)(\omega)\rightarrow v(x,t)(\omega)$ as $N\rightarrow\infty$ a.s., which implies convergence in law: $F_{v_N(x,t)}(v)\rightarrow F_{v(x,t)}(v)$ as $N\rightarrow\infty$, for all $v\in \mathbb{R}$ which is a point of continuity of $F_{v(x,t)}$. Here, $F$ refers to the distribution function. Since $f_{v_N(x,t)}$ is the density of $v_N(x,t)$,
\[ F_{v_N(x,t)}(v)=F_{v_N(x,t)}(v_0)+\int_{v_0}^v f_{v_N(x,t)}(w)\,\dif w. \]
If $v$ and $v_0$ are points of continuity of $F_{v(x,t)}$, taking limits when $N\rightarrow\infty$ we get
\[ F_{v(x,t)}(v)=F_{v(x,t)}(v_0)+\int_{v_0}^v g_{x,t}(w)\,\dif w. \]
This is justified by the Dominated Convergence Theorem, as 
\begin{equation}
 |f_{v_N(x,t)}(w)|\leq \mathbb{E}[|f_{A_1}(Z_N)||Y|]\leq \|f_{A_1}\|_{\leb^\infty(\mathbb{R})}\mathbb{E}[Y]\in\leb^1([v_0,v],\dif w). 
 \label{dct}
\end{equation}
As the points of discontinuity of $F_{v(x,t)}$ are countable and $F_{v(x,t)}$ is right-continuous, we obtain 
\[ F_{v(x,t)}(v)=F_{v(x,t)}(v_0)+\int_{v_0}^v g_{x,t}(w)\,\dif w, \]
for all $v_0$ and $v$ in $\mathbb{R}$. Thus, $g_{x,t}=f_{u(x,t)}$ is a density for $v(x,t)$, as wanted. 

\end{proof}

Our main objective will be to extend both Theorem \ref{teor2} and the new Theorem \ref{teorllei} to the solution of the randomized heat equation on an interval $[L_1,L_2]$ with random boundary conditions. 

Below, we state several key results that will be used in the subsequent development. The first result is the Random Variable Transformation technique, that will permit obtaining the density function in the forthcoming changes of variable (\ref{sol_u_w})--(\ref{psi_w}). Afterwards, we particularize the Random Variable Transformation technique to the case of the sum of two random variables, leading to the so-called density convolution formula (this is also a consequence of Remark \ref{xyabs}). Finally, we recall the Karhunen-Loève expansion for square integrable stochastic processes. The numerical examples will be applied to PDE problems (\ref{edp_determinista_comp}) in which we know the explicit Karhunen-Loève expansion of the initial condition stochastic process.

\begin{lemma}[Random Variable Transformation technique] \label{lema_abscont} \cite[Lemma 4.12]{llibre_powell} Let $X$ be an absolutely continuous random vector with density $f_X$ and with support $D_X$ contained in an open set $D\subseteq \mathbb{R}^n$. Let $g:D\rightarrow\mathbb{R}^n$ be a $C^1(D)$ function, injective on $D$ such that $Jg(x)\neq0$ for all $x\in D$ ($J$ stands for Jacobian). Let $h=g^{-1}:g(D)\rightarrow\mathbb{R}^n$. Let $Y=g(X)$ be a random vector. Then $Y$ is absolutely continuous with density
\[
 f_Y(y)=\begin{cases} f_X(h(y))|Jh(y)|, & \;y\in g(D), \\ 0, & \; y\notin g(D). \end{cases}  
\]
\end{lemma}

\begin{corollary}[Density Convolution Formula] \label{conv_dens} \cite[p.267]{bil} \cite[p.372]{pitman}
Let $X$ and $Y$ be independent and absolutely continuous random variables, with densities $f_X$ and $f_Y$, respectively. Then $X+Y$ is absolutely continuous, with density
\[ f_{X+Y}(z)=(f_X\ast f_Y)(z)=\int_\mathbb{R} f_X(z-w)f_Y(w)\,\dif w \]
(the symbol $\ast$ stands for the convolution).
\end{corollary}

\begin{lemma}[Karhunen-Lo\`{e}ve Theorem] \label{KLlemma} \cite[Th. 5.28]{llibre_powell} Consider a process $\{X(t):\,t\in\mathcal{T}\}$ in $\leb^2(\mathcal{T}\times\Omega)$. Then
\[ X(t,\omega)=\mu(t)+\sum_{j=1}^\infty \sqrt{\nu_j}\,\phi_j(t)\xi_j(\omega), \]
where the sum converges in $\leb^2(\mathcal{T}\times\Omega)$, $\mu(t)=\mathbb{E}[X(t)]$, $\{\phi_j\}_{j=1}^\infty$ is an orthonormal basis of $\leb^2(\mathcal{T})$, $\{(\nu_j,\phi_j)\}_{j=1}^\infty$ is the set of pairs of (nonnegative) eigenvalues and eigenfunctions of the operator
\begin{equation}
 \mathcal{C}:\leb^2(\mathcal{T})\rightarrow \leb^2(\mathcal{T}),\; \mathcal{C}f(t)=\int_{\mathcal{T}} \Cov[X(t),X(s)]f(s)\,\dif s, 
 \label{karhC}
\end{equation}
where $\Cov[\cdot,\cdot]$ is the covariance operator and $\{\xi_j\}_{j=1}^\infty$ is a sequence of random variables with zero expectation, unit variance and pairwise uncorrelated. Moreover, if $\{X(t):\,t\in\mathcal{T}\}$ is a Gaussian process, then $\{\xi_j\}_{j=1}^\infty$ are independent and Gaussian.
\end{lemma}

\subsection{Preliminaries on the $\leb^1(\Omega)$ and $\leb^2(\Omega)$ calculus} \label{preL} \ \\

In this section, we summarize the main results related to the so-called $\leb^p(\Omega)$
random calculus that will be required throughout our subsequent development. To read an exposition on $\leb^2(\Omega)$ calculus, see \cite[Ch.4]{Soong} and \cite[Ch.5]{llibre_powell}. In \cite{l2} authors combine $\leb^2(\Omega)$ and $\leb^4(\Omega)$ calculus, usually termed mean square and mean fourth random calculus, to solve random differential equations.

Let $X=\{X(t)(\omega):\,t\in I\subseteq\mathbb{R},\,\omega\in\Omega\}$ be a stochastic process, where $I$ is an open interval. Fix $1\leq p\leq\infty$. We say that $X$ is in $\leb^p(\Omega)$ if the random variable $X(t)$ belongs to $\leb^p(\Omega)$ for all $t\in I$. For such processes, we say that $X$ is continuous in the $\leb^p(\Omega)$ sense at $t_0\in I$ if 
\[ \lim_{h\rightarrow0} \|X(t_0+h)-X(t_0)\|_{\leb^p(\Omega)}=0. \]
We say that $X$ is continuous on $I$ in the $\leb^p(\Omega)$ sense if it is continuous in the $\leb^p(\Omega)$ sense at every $t_0\in I$.

It is said that $X$ is differentiable in the $\leb^p(\Omega)$ sense at $t_0\in I$ if there exists a random variable $X'(t_0)$ in $\leb^p(\Omega)$ such that
\[ \lim_{h\rightarrow0}\left\|\frac{X(t_0+h)-X(t_0)}{h}-X'(t_0)\right\|_{\leb^p(\Omega)}=0. \]
The random variable $X'(t_0)$ is called the $\leb^p(\Omega)$ derivative of $X(t)$ at $t_0$. We say that $X$ is differentiable on $I$ in the $\leb^p(\Omega)$ sense if it is differentiable in the $\leb^p(\Omega)$ sense at every $t_0\in I$. 

Differentiability in the $\leb^p(\Omega)$ sense at a point $t_0$ implies continuity in the $\leb^p(\Omega)$ sense at $t_0$, see \cite[p.95 (1)]{Soong}.

Recall Cauchy-Schwarz inequality: if $X$ and $Y$ are two random variables, then $\|XY\|_{\leb^1(\Omega)}\leq \|X\|_{\leb^2(\Omega)}\|Y\|_{\leb^2(\Omega)}$. As a consequence of this inequality, and with a similar prove to that of \cite[Lemma 3.14]{l2}, we have the following two results:

\begin{proposition} \label{prodcont}
Let $X=\{X(t)(\omega):\,t\in I,\,\omega\in\Omega\}$ and $Y=\{Y(t)(\omega):\,t\in I,\,\omega\in\Omega\}$ be two stochastic processes. Suppose that they are continuous at $t_0\in I$ in the $\leb^2(\Omega)$ sense. Then $XY=\{X(t)(\omega)Y(t)(\omega):\,t\in I,\,\omega\in\Omega\}$ is continuous at $t_0$ in the $\leb^1(\Omega)$ sense.
\end{proposition}

\begin{proposition} \label{producte}
Let $X=\{X(t)(\omega):\,t\in I,\,\omega\in\Omega\}$ and $Y=\{Y(t)(\omega):\,t\in I,\,\omega\in\Omega\}$ be two stochastic processes. Suppose that they are differentiable at $t_0\in I$ in the $\leb^2(\Omega)$ sense. Then $XY=\{X(t)(\omega)Y(t)(\omega):\,t\in I,\,\omega\in\Omega\}$ is differentiable at $t_0$ in the $\leb^1(\Omega)$ sense, with $(XY)'(t_0)=X'(t_0)Y(t_0)+X(t_0)Y'(t_0)$. 
\end{proposition}

Another useful result is the following, see \cite[p.97]{Soong}:

\begin{proposition} \label{derexp}
Let $\{X_n\}_{n=1}^\infty$ be a sequence of random variables that converges in $\leb^1(\Omega)$ to the random variable $X$. Then $\lim_{n\rightarrow\infty} \mathbb{E}[X_n]=\mathbb{E}[X]$.

As a consequence, if $X=\{X(t)(\omega):\,t\in I,\,\omega\in\omega\}$ is a stochastic process that is differentiable in the $\leb^1(\Omega)$ sense at $t_0\in I$, then there exists
\[ \frac{\dif}{\dif t}\mathbb{E}[X(t_0)]=\mathbb{E}[X'(t_0)]. \]
\end{proposition}

A similar result but in terms of continuity holds:

\begin{proposition} \label{contexp}
If $X=\{X(t)(\omega):\,t\in I,\,\omega\in\omega\}$ is a stochastic process that is continuous in the $\leb^1(\Omega)$ sense at $t_0\in I$, then $\mathbb{E}[X(t)]$ is continuous at $t_0$.

If $X$ is continuous in the $\leb^2(\Omega)$ sense at $t_0\in I$, then $\|X(t)\|_{\leb^2(\Omega)}$ is continuous at $t_0$.
\end{proposition}

\section{Solution stochastic process to the randomized heat equation with non-homogeneous boundary conditions}

Consider the general form of the deterministic heat equation,
\begin{equation} \begin{cases} u_t=\alpha^2 u_{xx},\quad L_1<x<L_2,\;t>0, \\ u(L_1,t)=A,\;u(L_2,t)=B,\quad t\geq0, \\ u(x,0)=\phi(x),\quad L_1\leq x\leq L_2. \end{cases} 
\label{edp_determinista_comp}
\end{equation}
This is a generalization of the PDE problem (\ref{edp_determinista_homo}) studied in \cite{jjcm}. The diffusion coefficient is $\alpha^2>0$, the boundary conditions are constants $A$ and $B$ and the initial condition is given by the function $\phi(x)$.

First of all, we relate the solution $u(x,t)$ of (\ref{edp_determinista_comp}) and the solution $v(x,t)$ of (\ref{edp_determinista_homo}). We need to transform the domain from $[0,1]$ to $[L_1,L_2]$ and we need to translate $v$ so that the new boundary conditions hold. The relation of both solutions is
\begin{equation}
 u(x,t)=v\left(\frac{x-L_1}{L_2-L_1},t\right)+\frac{x-L_1}{L_2-L_1}B+\frac{L_2-x}{L_2-L_1}A, 
 \label{sol_u}
\end{equation}
where $x\in [L_1,L_2]$ and $t\geq0$, and $v(y,t)$ is the solution of (\ref{edp_determinista_homo}) with diffusion coefficient
\begin{equation} \beta^2=\frac{\alpha^2}{(L_2-L_1)^2} 
 \label{beta}
\end{equation}
and initial condition
\begin{equation} \psi(y)=\varphi(L_1+y(L_2-L_1)),\;\;\varphi(x)=\phi(x)-\frac{x-L_1}{L_2-L_1}B-\frac{L_2-x}{L_2-L_1}A, 
 \label{psi}
\end{equation}
for $y\in [0,1]$. Theorem \ref{c1} generalizes to problem (\ref{edp_determinista_comp}).

\begin{theorem}
If $\phi$ is continuous on $[L_1,L_2]$, piecewise $C^1$ on $[L_1,L_2]$, $\phi(L_1)=A$ and $\phi(L_2)=B$, then $u(x,t)$ is continuous on $[L_1,L_2]\times [0,\infty)$, is of class $C^{2,1}$ on $(L_1,L_2)\times(0,\infty)$ and is a classical solution of (\ref{edp_determinista_comp}).
\end{theorem}
\begin{proof}
From (\ref{psi}), we have that $\psi$ is continuous on $[0,1]$ and piecewise $C^1$ on $[0,1]$. We also have $\psi(0)=\varphi(L_1)=\phi(L_1)-A=0$ and $\psi(1)=\varphi(L_2)=\phi(L_2)-B=0$. By Theorem \ref{c1}, $v$ is continuous on $[0,1]\times [0,\infty)$, is of class $C^{2,1}$ on $(0,1)\times (0,\infty)$ and is a classical solution of (\ref{edp_determinista_homo}). From (\ref{sol_u}), we deduce that $u$ is continuous on $[L_1,L_2]\times [0,\infty)$ and of class $C^{2,1}$ on $(L_1,L_2)\times(0,\infty)$. We have
\begin{align}
 u_t(x,t)= {} & v_t\left(\frac{x-L_1}{L_2-L_1},t\right)=\beta^2v_{xx}\left(\frac{x-L_1}{L_2-L_1},t\right) \nonumber \\
= {} & \frac{\alpha^2}{(L_2-L_1)^2}v_{xx}\left(\frac{x-L_1}{L_2-L_1},t\right)=\alpha^2 u_{xx}(x,t), \label{utuxx}
\end{align}
\begin{equation}
 u(L_1,t)=v(0,t)+A=A,\quad u(L_2,t)=v(1,t)+B=B
 \label{cb}
\end{equation}
and
\begin{align}
 u(x,0)= {} & v\left(\frac{x-L_1}{L_2-L_1},0\right)+\frac{x-L_1}{L_2-L_1}B+\frac{L_2-x}{L_2-L_1}A \nonumber \\
= {} & \psi\left(\frac{x-L_1}{L_2-L_1}\right)+\frac{x-L_1}{L_2-L_1}B+\frac{L_2-x}{L_2-L_1}A \nonumber \\
= {} & \varphi(x)+\frac{x-L_1}{L_2-L_1}B+\frac{L_2-x}{L_2-L_1}A=\phi(x).  \label{ux0}
\end{align}

\end{proof}

We reformulate the heat equation (\ref{edp_determinista_comp}) in a random setting. Given a complete probability space $(\Omega,\mathcal{F},\mathbb{P})$, we will consider the diffusion coefficient $\alpha^2(\omega)$ and the boundary conditions $A(\omega)$ and $B(\omega)$ as random variables, and the initial condition as a stochastic process
\[ \phi=\{\phi(x)(\omega):\,x\in [L_1,L_2],\,\omega\in\Omega\} \]
in the underlying probability space. In this random setting, the change of variable (\ref{sol_u})--(\ref{psi}) is understood pointwise in $\omega\in\Omega$:
\begin{equation}
 u(x,t)(\omega)=v\left(\frac{x-L_1}{L_2-L_1},t\right)(\omega)+\frac{x-L_1}{L_2-L_1}B(\omega)+\frac{L_2-x}{L_2-L_1}A(\omega), 
 \label{sol_u_w}
\end{equation}
where $x\in [L_1,L_2]$ and $t\geq0$, and $v(y,t)(\omega)$ is the solution stochastic process of (\ref{edp_determinista_homo}) given by (\ref{sol_homo_random})--(\ref{An}) with random diffusion coefficient
\begin{equation} 
 \beta^2(\omega)=\frac{\alpha^2(\omega)}{(L_2-L_1)^2} 
 \label{beta_w}
\end{equation}
and random initial condition
\begin{equation} \psi(y)(\omega)=\varphi(L_1+y(L_2-L_1))(\omega),\;\;\varphi(x)(\omega)=\phi(x)(\omega)-\frac{x-L_1}{L_2-L_1}B(\omega)-\frac{L_2-x}{L_2-L_1}A(\omega), 
 \label{psi_w}
\end{equation}
for $y\in [0,1]$.

We want to study in which sense the stochastic process $u(x,t)(\omega)$ given by (\ref{sol_u_w}) is a rigorous solution to the randomized heat equation (\ref{edp_determinista_comp}). Next Theorem \ref{st_sol_2} generalizes Theorem \ref{st_sol}. Moreover, uniqueness is proved, which is a novelty compared with \cite{jjcm}.

\begin{theorem} \label{st_sol_2}
The following statements hold:\\
\vspace{-0.5cm}
\begin{itemize}
\item[i)] a.s. solution: Suppose that $\phi\in \leb^2([L_1,L_2]\times\Omega)$ and $A,B\in\leb^2(\Omega)$. Then
\[ u_t(x,t)(\omega)=\alpha^2(\omega)\,u_{xx}(x,t)(\omega) \]
a.s. for $x\in (L_1,L_2)$ and $t>0$, where the derivatives are understood in the classical sense; $u(L_1,t)(\omega)=A(\omega)$ and $u(L_2,t)(\omega)=B(\omega)$ a.s. for $t\geq0$; and $u(x,0)(\omega)=\phi(x)(\omega)$ a.s. for a.e. $x\in [L_1,L_2]$. Moreover, the process $u(x,t)(\omega)$ satisfying these conditions is unique.
\item[ii)] $\leb^2$ solution: Suppose that $\phi\in \leb^2([L_1,L_2]\times\Omega)$, $A,B\in\leb^2(\Omega)$ and $0<a\leq\alpha^2(\omega)\leq b$, a.e. $\omega\in\Omega$, for certain $a,b\in\mathbb{R}$. Then
\[ u_t(x,t)(\omega)=\alpha^2(\omega)\,u_{xx}(x,t)(\omega) \]
a.s. for $x\in (L_1,L_2)$ and $t>0$, where the derivatives are understood in the mean square sense; $u(L_1,t)(\omega)=A(\omega)$ and $u(L_2,t)(\omega)=B(\omega)$ a.s. for $t\geq0$; and $u(x,0)(\omega)=\phi(x)(\omega)$ a.s. for a.e. $x\in [L_1,L_2]$. Moreover, the process $u(x,t)(\omega)$ satisfying these conditions is unique.
\end{itemize}
\end{theorem}
\begin{proof} \ \\
\begin{itemize}
\item[i)] a.s. solution: By (\ref{psi_w}) and the triangular inequality,
\small
\begin{align}
 \|\psi {} & \|_{\leb^2([0,1]\times \Omega)}=\left(\mathbb{E}\left[\int_0^1 \psi(y)^2\,\dif y\right]\right)^{\frac12}=\left(\mathbb{E}\left[\int_0^1 \varphi(L_1+y(L_2-L_1))^2\,\dif y\right]\right)^{\frac12} \nonumber \\
= {} & \frac{1}{\sqrt{L_2-L_1}} \left(\mathbb{E}\left[\int_{L_1}^{L_2} \varphi(x)^2\,\dif x\right]\right)^{\frac12} = \frac{1}{\sqrt{L_2-L_1}} \|\varphi\|_{\leb^2([L_1,L_2]\times \Omega)} \nonumber \\
\leq {} & \frac{1}{\sqrt{L_2-L_1}}\left(\|\phi\|_{\leb^2([L_1,L_2]\times \Omega)}+\frac{x-L_1}{L_2-L_1}\|B\|_{\leb^2([L_1,L_2]\times \Omega)}+\frac{L_2-x}{L_2-L_1}\|A\|_{\leb^2([L_1,L_2]\times \Omega)}\right) \nonumber \\
= {} & \frac{1}{\sqrt{L_2-L_1}}\|\phi\|_{\leb^2([L_1,L_2]\times \Omega)}+\frac{x-L_1}{L_2-L_1}\|B\|_{\leb^2(\Omega)}+\frac{L_2-x}{L_2-L_1}\|A\|_{\leb^2(\Omega)}<\infty. \label{psiL2}
\end{align}
\normalsize
Then $\psi\in\leb^2([0,1]\times \Omega)$. By Theorem \ref{st_sol}, $v_t(x,t)(\omega)=\beta^2(\omega)\,v_{xx}(x,t)(\omega)$ a.s. for $x\in (0,1)$ and $t>0$, where the derivatives are understood in the classical sense; $v(0,t)(\omega)=v(1,t)(\omega)=0$ a.s. for $t\geq0$; and $v(x,0)(\omega)=\psi(x)(\omega)$ a.s. for a.e. $x\in [0,1]$. As we did in (\ref{utuxx})--(\ref{ux0}), we derive that: $u_t(x,t)(\omega)=\alpha^2(\omega)\,u_{xx}(x,t)(\omega)$ a.s. for $x\in (L_1,L_2)$ and $t>0$, where the derivatives are understood in the classical sense; $u(L_1,t)(\omega)=A(\omega)$ and $u(L_2,t)(\omega)=B(\omega)$ a.s. for $t\geq0$; and $u(x,0)(\omega)=\phi(x)(\omega)$ a.s. for a.e. $x\in [L_1,L_2]$.

Uniqueness follows from the so-called energy method \cite[p.30--31]{salsa}.

\item[ii)] $\leb^2$ solution: By (\ref{psiL2}), $\psi\in\leb^2([0,1]\times \Omega)$. By Theorem \ref{st_sol}, $v_t(x,t)(\omega)=\beta^2(\omega)\,v_{xx}(x,t)(\omega)$ a.s. for $x\in (0,1)$ and $t>0$, where the derivatives are understood in the mean square sense; $v(0,t)(\omega)=v(1,t)(\omega)=0$ a.s. for $t\geq0$; and $v(x,0)(\omega)=\psi(x)(\omega)$ a.s. for a.e. $x\in [0,1]$. Equalities (\ref{cb}) and (\ref{ux0}) hold in this setting as well. We need to check that the chain rule applied in (\ref{utuxx}) holds for the $\leb^2$ derivative too. Notice that, by the structure of relation (\ref{sol_u_w}), it is enough to check that: if $X(x)(\omega)=Y(cx+d)(\omega)+E(\omega)x+F(\omega)$, where $c,d\in\mathbb{R}$ and $Y$ differentiable in the $\leb^2$ sense, then $X$ is differentiable in the $\leb^2$ sense and $X'(x)(\omega)=cY'(cx+d)(\omega)+E(\omega)$. This is a consequence of the following two limits:
\small
\begin{align*} {} & \lim_{h\rightarrow0} \left\|\frac{Y(c(x+h)+d)-Y(cx+d)}{h}-cY'(cx+d)\right\|_{\leb^2(\Omega)} \\
= {} & |c|\lim_{h\rightarrow0} \left\|\frac{Y(cx+d+ch)-Y(cx+d)}{ch}-Y'(cx+d)\right\|_{\leb^2(\Omega)}=0, 
\end{align*}
\[ \lim_{h\rightarrow0} \left\|\frac{(E(\omega)(x+h)+F(\omega))-(E(\omega)x+F(\omega))}{h}-E(\omega)\right\|_{\leb^2(\Omega)}=0. \]
\normalsize
Thus, the differentiations in (\ref{utuxx}) are justified in the $\leb^2(\Omega)$ sense, and the conclusion of the theorem follows: $u_t(x,t)(\omega)=\alpha^2(\omega)\,u_{xx}(x,t)(\omega)$ a.s. for $x\in (L_1,L_2)$ and $t>0$, where the derivatives are understood in the mean square sense; $u(L_1,t)(\omega)=A(\omega)$ and $u(L_2,t)(\omega)=B(\omega)$ a.s. for $t\geq0$; and $u(x,0)(\omega)=\phi(x)(\omega)$ a.s. for a.e. $x\in [L_1,L_2]$.

To show uniqueness, we try to adapt the energy method \cite[p.30--31]{salsa} to this setting. We prove the following: if $u$ is $C^{2,1}((L_1,L_2)\times (0,\infty))$ in the sense of $\leb^2(\Omega)$, with continuous partial derivatives on $[L_1,L_2]\times [0,\infty)$ in the sense of $\leb^2(\Omega)$, $u_t=\alpha^2u_{xx}$ on $(L_1,L_2)\times (0,\infty)$, $u(L_1,t)=u(L_2,t)=0$ a.s. on $[0,\infty)$ and $u(x,0)=0$ a.s. at a.e. $x\in [L_1,L_2]$, then $u(x,t)=0$ a.s. for all $x\in [L_1,L_2]$ and $t\geq0$. From this fact, uniqueness will follow.

Let $I(t)=\int_{L_1}^{L_2} \mathbb{E}[u(x,t)^2]\,\dif x$. Fixed $t\geq0$, as a consequence of the continuity of $u(\cdot,t)$ in the $\leb^2(\Omega)$ sense and Proposition \ref{contexp}, the real map $x\in [L_1,L_2]\mapsto \mathbb{E}[u(x,t)^2]$ is continuous and $I(t)$ is well-defined. Fixed $x$, as $u(x,\cdot)$ is differentiable in the $\leb^2(\Omega)$ sense, by Proposition \ref{producte} and Proposition \ref{derexp} we have:
\[ \frac{\partial}{\partial t}\mathbb{E}[u(x,t)^2]=\mathbb{E}\bigg[\frac{\partial}{\partial t}\left(u(x,t)^2\right)\bigg]=2\mathbb{E}[u(x,t)u_t(x,t)], \]
where the partial derivative $\frac{\partial}{\partial t}$ inside the expectation operator must be understood on the $\leb^1(\Omega)$ sense.
Then, using Cauchy-Schwarz inequality,
$$
\left|\frac{\partial}{\partial t}\mathbb{E}\left[u(x,t)^2\right]\right|\leq 2\mathbb{E}[|u(x,t)||u_t(x,t)|]\leq 2\|u(x,t)\|_{\leb^2(\Omega)}\|u_t(x,t)\|_{\leb^2(\Omega)}. 
$$
As both $u(x,t)$ and $u_t(x,t)$ are continuous on $[L_1,L_2]\times [0,\infty)$ in the $\leb^2(\Omega)$ sense, by Proposition \ref{contexp} both $\|u(x,t)\|_{\leb^2(\Omega)}$ and $\|u_t(x,t)\|_{\leb^2(\Omega)}$ are continuous on $[L_1,L_2]\times [0,\infty)$ in the classical sense. Fix $t_0>0$ and $\delta>0$ small. Then there is a constant $C>0$ such that $\|u(x,t)\|_{\leb^2(\Omega)}\leq C$ and $\|u_t(x,t)\|_{\leb^2(\Omega)}\leq C$ for all $x\in [L_1,L_2]$ and $t\in [t_0-\delta,t_0+\delta]$, by continuity in the classical sense. Thus,
\[ \left|\frac{\partial}{\partial t}\mathbb{E}\left[u(x,t)^2\right]\right|\leq 2C^2\in\leb^1([L_1,L_2],\dif x), \]
for all $x\in [L_1,L_2]$ and $t\in [t_0-\delta,t_0+\delta]$. This permits differentiating under the Lebesgue integral sign at $t_0$ \cite[Th. 10.39]{apostol}:
\begin{align*}
 I'(t_0)= {} & \frac{\partial}{\partial t}\left(\int_{L_1}^{L_2} \mathbb{E}\left[u(x,t)^2\right]\,\dif x\right)\bigg|_{t=t_0}=\int_{L_1}^{L_2} \frac{\partial}{\partial t}\left(\mathbb{E}\left[u(x,t)^2\right]\right)\big|_{t=t_0}\,\dif x \\
= {} & 2\int_{L_1}^{L_2} \mathbb{E}[u(x,t_0)u_t(x,t_0)]\,\dif x. 
\end{align*}
Now we use the arbitrariness of $t_0$ and the fact that $u$ solves the heat equation:
\[ I'(t)=2\int_{L_1}^{L_2} \mathbb{E}[u(x,t)u_t(x,t)]\,\dif x=2\int_{L_1}^{L_2} \mathbb{E}[\alpha^2u(x,t)u_{xx}(x,t)]\,\dif x. \]
As $u(\cdot,t)$ and $u_x(\cdot,t)$ are differentiable in the $\leb^2(\Omega)$ sense, by Proposition \ref{producte} the product $u(\cdot,t)u_x(\cdot,t)$ is differentiable in the $\leb^1(\Omega)$ sense, with derivative $(u(x,t)u_x(x,t))_x=u(x,t)u_{xx}(x,t)+u_x(x,t)^2$. Since $\alpha^2$ is bounded above, $\alpha^2u(\cdot,t)u_x(\cdot,t)$ is differentiable in the $\leb^1(\Omega)$ sense, having derivative $(\alpha^2u(x,t)u_x(x,t))_x=\alpha^2u(x,t)u_{xx}(x,t)+\alpha^2u_x(x,t)^2$. Thereby,
\[ I'(t)=2\int_{L_1}^{L_2} \mathbb{E}[(\alpha^2u(x,t)u_{x}(x,t))_x]\,\dif x-2\int_{L_1}^{L_2} \mathbb{E}[\alpha^2u_{x}(x,t)^2]\,\dif x. \]
By Proposition \ref{derexp}, Barrow's rule and the boundary conditions, the first integral is $0$: 
\begin{align*}
 \int_{L_1}^{L_2} \mathbb{E}[(\alpha^2u(x,t)u_{x}(x,t))_x]{} &\,\dif x=\int_{L_1}^{L_2}\frac{\partial}{\partial x}\mathbb{E}[\alpha^2u(x,t)u_{x}(x,t)]\,\dif x \\
= {} & \mathbb{E}[\alpha^2u(L_2,t)u_x(L_2,t)]-\mathbb{E}[\alpha^2u(L_1,t)u_x(L_1,t)]=0. 
\end{align*}
Barrow's rule is justified as follows: we have, by previous computations,
\begin{align*}
\partial_x \mathbb{E}[\alpha^2u(x,t)u_{x}(x,t)] ={} &\mathbb{E}[(\alpha^2 u(x,t)u_x(x,t))_x]\\
={} &\mathbb{E}[\alpha^2u(x,t)u_{xx}(x,t)]+\mathbb{E}[\alpha^2u_x(x,t)^2]. 
\end{align*}
By Proposition \ref{prodcont} and the boundedness of $\alpha^2$, both $\alpha^2u(\cdot,t)u_{xx}(\cdot,t)$ and $\alpha^2u_x(\cdot,t)^2$ are continuous in the $\leb^1(\Omega)$ sense. So by Proposition \ref{contexp}, both $\mathbb{E}[\alpha^2u(\cdot,t)u_{xx}(\cdot,t)]$ and $\mathbb{E}[\alpha^2u_x(\cdot,t)^2]$ are continuous. Then $\partial_x \mathbb{E}[\alpha^2u(\cdot,t)u_{x}(\cdot,t)]$ is continuous on $[L_1,L_2]$ and Barrow's rule is applicable.

It follows $I'(t)=-2\int_{L_1}^{L_2} \mathbb{E}[\alpha^2u_{x}(x,t)^2]\,\dif x\leq 0$. This tells us that $I(t)$ is decreasing on $[0,\infty)$, which implies $I(t)\leq I(0)=\int_{L_1}^{L_2} \mathbb{E}[u(x,0)^2]\,\dif x=0$. Hence, $I(t)=0$. As $\mathbb{E}[u(\cdot,t)^2]$ is continuous, because $u(\cdot,t)$ is continuous in the $\leb^2(\Omega)$ sense and Proposition \ref{contexp}, we derive that $E[u(x,t)^2]=0$ for all $x\in [L_1,L_2]$ and $t\geq0$. Then $u(x,t)=0$ a.s., for every $x\in [L_1,L_2]$ and $t\geq0$. This concludes the proof.

\end{itemize}

\end{proof}

\section{Approximation of the probability density function of the solution stochastic process}

The main goal of this paper is to approximate the probability density function of the solution stochastic process $u(x,t)(\omega)$ given by (\ref{sol_u_w}), which solves the random heat equation (\ref{edp_determinista_comp}). We will use Theorem \ref{teor2}, Theorem \ref{teorllei} and Lemma \ref{lema_abscont}.

Assume that $v(y,t)(\omega)$, $A(\omega)$ and $B(\omega)$ are absolutely continuous and independent random variables. Applying Lemma \ref{lema_abscont},
\[ f_{\frac{x-L_1}{L_2-L_1}B}(b)=f_B\left(\frac{L_2-L_1}{x-L_1}b\right)\frac{L_2-L_1}{x-L_1} \]
and 
\[ f_{\frac{L_2-x}{L_2-L_1}A}(a)=f_A\left(\frac{L_2-L_1}{L_2-x}a\right)\frac{L_2-L_1}{L_2-x}. \]
By Corollary \ref{conv_dens}, the probability density function of a sum of two independent and absolutely continuous random variables is given by the convolution of their probability density functions. Thereby, from (\ref{sol_u_w}),
\small
\begin{align*}
{} & f_{u(x,t)}(u)=\int_{\mathbb{R}}\int_{\mathbb{R}} f_{v\left(\frac{x-L_1}{L_2-L_1},t\right)}(u-b-a)f_{\frac{x-L_1}{L_2-L_1}B}(b)f_{\frac{L_2-x}{L_2-L_1}A}(a)\,\dif a\, \dif b \\
= {} & \int_{\mathbb{R}}\int_{\mathbb{R}} f_{v\left(\frac{x-L_1}{L_2-L_1},t\right)}(u-b-a)f_B\left(\frac{L_2-L_1}{x-L_1}b\right)\frac{L_2-L_1}{x-L_1}f_A\left(\frac{L_2-L_1}{L_2-x}a\right)\frac{L_2-L_1}{L_2-x}\,\dif a\,\dif b. 
\end{align*}
\normalsize

Define a new truncation
\begin{equation}
 u_N(x,t)(\omega)=v_N\left(\frac{x-L_1}{L_2-L_1},t\right)(\omega)+\frac{x-L_1}{L_2-L_1}B(\omega)+\frac{L_2-x}{L_2-L_1}A(\omega), 
 \label{sol_u_w_N}
\end{equation}
where $x\in [L_1,L_2]$ and $t\geq0$ and $v_N$ is the truncation (\ref{vN}). If $v_N(y,t)(\omega)$, $A(\omega)$ and $B(\omega)$ are absolutely continuous and independent random variables, by Corollary \ref{conv_dens} again, 
\small
\begin{align}
{} & f_{u_N(x,t)}(u)=\int_{\mathbb{R}}\int_{\mathbb{R}} f_{v_N\left(\frac{x-L_1}{L_2-L_1},t\right)}(u-b-a)f_{\frac{x-L_1}{L_2-L_1}B}(b)f_{\frac{L_2-x}{L_2-L_1}A}(a)\,\dif a\, \dif b \nonumber \\
= {} & \int_{\mathbb{R}}\int_{\mathbb{R}} f_{v_N\left(\frac{x-L_1}{L_2-L_1},t\right)}(u-b-a)f_B\left(\frac{L_2-L_1}{x-L_1}b\right)\frac{L_2-L_1}{x-L_1}f_A\left(\frac{L_2-L_1}{L_2-x}a\right)\frac{L_2-L_1}{L_2-x}\,\dif a\,\dif b.  \label{fN}
\end{align}
\normalsize

Intuitively, we should be able to set conditions under which 
\[ \lim_{N\rightarrow\infty} f_{u_N(x,t)}(u)=f_{u(x,t)}(u), \]
as an application of Theorem \ref{teor2} or of Theorem \ref{teorllei}. This fact is formalized in the following two theorems.

\begin{theorem} \label{super}
Let the random initial condition $\{\phi(x):\,L_1\leq x\leq L_2\}$ be a process in $\leb^2([L_1,L_2]\times \Omega)$. Let the random boundary conditions $A$ and $B$ belong to $\leb^2(\Omega)$. Suppose that $\alpha^2$, $A_1$, $(A_2,\ldots,A_N)$, $A$ and $B$ are independent and absolutely continuous, for $N\geq2$ (recall that $A_n$ is defined in (\ref{An}) as the random Fourier coefficient of $\psi$, where $\psi$ is defined from $\phi$ in relation (\ref{psi_w})). Suppose that the probability density function $f_{A_1}$ is Lipschitz on $\mathbb{R}$. Assume that 
\[ \sum_{n=m}^\infty \|\e^{-(n^2-2)\pi^2\alpha^2 t/(L_2-L_1)^2}\|_{\leb^1(\Omega)}<\infty, \] 
for certain $m\in\mathbb{N}$. Then the sequence
\small
\[ f_{u_N(x,t)}(u)=\int_{\mathbb{R}}\int_{\mathbb{R}} f_{v_N\left(\frac{x-L_1}{L_2-L_1},t\right)}(u-b-a)f_B\left(\frac{L_2-L_1}{x-L_1}b\right)\frac{L_2-L_1}{x-L_1}f_A\left(\frac{L_2-L_1}{L_2-x}a\right)\frac{L_2-L_1}{L_2-x}\,\dif a\,\dif b, \]
\normalsize
where $f_{v_N}$ is the density defined by (\ref{fr}), converges in $\leb^\infty(\mathbb{R})$ to the density $f_{u(x,t)}(u)$ of the solution stochastic process $u(x,t)(\omega)$ to the randomized heat equation (\ref{edp_determinista_comp}), for $L_1<x<L_2$ and $t>0$.
\end{theorem}
\begin{proof}
Since $\phi\in \leb^2([L_1,L_2]\times \Omega)$ and by (\ref{psiL2}), $\psi\in \leb^2([0,1]\times \Omega)$. By hypothesis, we also have that $\beta^2=\alpha^2/(L_2-L_1)^2$, $A_1$ and $(A_2,\ldots,A_N)$ are independent and absolutely continuous, for $N\geq2$, and $\sum_{n=m}^\infty \|\e^{-(n^2-2)\pi^2\beta^2 t}\|_{\leb^1(\Omega)}<\infty$. Thus, the hypotheses of Theorem \ref{teor2} hold.

Since $\alpha^2$, $A_1$, $(A_2,\ldots,A_N)$, $A$ and $B$ are independent, from (\ref{sol_homo_random}) we derive that $v(y,t)$, $A$ and $B$ are independent. Indeed, from the independence of $\beta^2=\alpha^2/(L_2-L_1)^2$, $A_1$, $(A_2,\ldots,A_N)$, $A$ and $B$, one has independence of $(\beta^2,A_1,\ldots,A_N)$, $A$ and $B$. Fixed $0<y<1$ and $t>0$, the random variable $v_N(y,t)(\omega)$ can be written as $g(\beta^2(\omega),A_1(\omega),\ldots,A_N(\omega))$, for a Borel measurable map $g:\mathbb{R}^{n+1}\rightarrow\mathbb{R}$. Then $v_N(y,t)$, $A$ and $B$ are independent. By Theorem \ref{st_sol} i), $v_N(y,t)\rightarrow v(y,t)$ a.s. as $N\rightarrow \infty$. Then $(v_N(y,t),A,B)\rightarrow (v(y,t),A,B)$ a.s. as $N\rightarrow \infty$. Denote by $\varphi$ the characteristic function. By L\'{e}vy's continuity theorem \cite[Ch.18]{martingales} and the independence, for $v,a,b \in \mathbb{R}$, 
\begin{align*}
\varphi_{(v(y,t),A,B)}(v,a,b)= {} & \lim_{N\rightarrow \infty}\varphi_{(v_N(y,t),A,B)}(v,a,b)=\lim_{N\rightarrow \infty} \varphi_{v_N(y,t)}(v)\varphi_A(a)\varphi_B(b) \\
= {} & \varphi_{v(y,t)}(v)\varphi_A(a)\varphi_B(b).
\end{align*}
By \cite[Th. 2.1]{llarg}, $v(y,t)$, $A$ and $B$ are independent.

As a consequence,
\small
\[ f_{u_N(x,t)}(u)=\int_{\mathbb{R}}\int_{\mathbb{R}} f_{v_N\left(\frac{x-L_1}{L_2-L_1},t\right)}(u-b-a)f_B\left(\frac{L_2-L_1}{x-L_1}b\right)\frac{L_2-L_1}{x-L_1}f_A\left(\frac{L_2-L_1}{L_2-x}a\right)\frac{L_2-L_1}{L_2-x}\,\dif a\,\dif b \]
\normalsize
and
\small
\[ f_{u(x,t)}(u)=\int_{\mathbb{R}}\int_{\mathbb{R}} f_{v\left(\frac{x-L_1}{L_2-L_1},t\right)}(u-b-a)f_B\left(\frac{L_2-L_1}{x-L_1}b\right)\frac{L_2-L_1}{x-L_1}f_A\left(\frac{L_2-L_1}{L_2-x}a\right)\frac{L_2-L_1}{L_2-x}\,\dif a\,\dif b. \]
\normalsize

We have the following estimates:
\footnotesize
\begin{align*}
{} & |f_{u(x,t)}(u)-f_{u_N(x,t)}(u)| \\
={} &  \bigg| \int_{\mathbb{R}}\int_{\mathbb{R}} f_{v\left(\frac{x-L_1}{L_2-L_1},t\right)}(u-b-a)f_B\left(\frac{L_2-L_1}{x-L_1}b\right)\frac{L_2-L_1}{x-L_1}f_A\left(\frac{L_2-L_1}{L_2-x}a\right)\frac{L_2-L_1}{L_2-x}\,\dif a\,\dif b \\
- {} & \int_{\mathbb{R}}\int_{\mathbb{R}} f_{v_N\left(\frac{x-L_1}{L_2-L_1},t\right)}(u-b-a)f_B\left(\frac{L_2-L_1}{x-L_1}b\right)\frac{L_2-L_1}{x-L_1}f_A\left(\frac{L_2-L_1}{L_2-x}a\right)\frac{L_2-L_1}{L_2-x}\,\dif a\,\dif b \bigg| \\
\leq {} & \frac{(L_2-L_1)^2}{(x-L_1)(L_2-x)}\int_\mathbb{R}\int_\mathbb{R}\bigg\{ f_B\left(\frac{L_2-L_1}{x-L_1}b\right)f_A\left(\frac{L_2-L_1}{L_2-x}a\right) \\
\cdot {} & \left|f_{v\left(\frac{x-L_1}{L_2-L_1},t\right)}(u-b-a)-f_{v_N\left(\frac{x-L_1}{L_2-L_1},t\right)}(u-b-a)\right|\bigg\}\,\dif a\,\dif b.
\end{align*}
\normalsize
By (\ref{rate}),
\begin{align*}
{} & \left|f_{v\left(\frac{x-L_1}{L_2-L_1},t\right)}(u-b-a)-f_{v_N\left(\frac{x-L_1}{L_2-L_1},t\right)}(u-b-a)\right| \\
\leq {} & \frac{2\|\psi\|_{\leb^2([0,1]\times\Omega)}L}{\sin^2\left(\pi \frac{x-L_1}{L_2-L_1}\right)} \sum_{n=N+1}^\infty \|\e^{-(n^2-2)\pi^2\alpha^2 t/(L_2-L_1)^2}\|_{\leb^1(\Omega)}, 
\end{align*}
where $L$ is the Lipschitz constant of $f_{A_1}$.
Then,
\small
\begin{align}
{} & |f_{u_N(x,t)}(u)-f_{u(x,t)}(u)| \nonumber \\
\leq {} & \frac{(L_2-L_1)^2}{(x-L_1)(L_2-x)}\frac{2\|\psi\|_{\leb^2([0,1]\times\Omega)}L}{\sin^2\left(\pi \frac{x-L_1}{L_2-L_1}\right)} \left(\sum_{n=N+1}^\infty \|\e^{-(n^2-2)\pi^2\alpha^2 t/(L_2-L_1)^2}\|_{\leb^1(\Omega)}\right) \nonumber \\
\cdot {} & \int_\mathbb{R}\int_\mathbb{R}f_B\left(\frac{L_2-L_1}{x-L_1}b\right)f_A\left(\frac{L_2-L_1}{L_2-x}a\right)\,\dif a\,\dif b \nonumber \\
= {} & \frac{(L_2-L_1)^2}{(x-L_1)(L_2-x)}\frac{2\|\psi\|_{\leb^2([0,1]\times\Omega)}L}{\sin^2\left(\pi \frac{x-L_1}{L_2-L_1}\right)} \left(\sum_{n=N+1}^\infty \|\e^{-(n^2-2)\pi^2\alpha^2 t/(L_2-L_1)^2}\|_{\leb^1(\Omega)}\right) \nonumber \\
\cdot {} & \left(\int_\mathbb{R}f_B\left(\frac{L_2-L_1}{x-L_1}b\right)\,\dif b\right)\left(\int_\mathbb{R}f_A\left(\frac{L_2-L_1}{L_2-x}a\right)\,\dif a\right) \nonumber \\
= {} & \|f_A\|_{\leb^1(\mathbb{R})}\|f_B\|_{\leb^1(\mathbb{R})}\frac{2\|\psi\|_{\leb^2([0,1]\times\Omega)}L}{\sin^2\left(\pi \frac{x-L_1}{L_2-L_1}\right)} \sum_{n=N+1}^\infty \|\e^{-(n^2-2)\pi^2\alpha^2 t/(L_2-L_1)^2}\|_{\leb^1(\Omega)}. \label{anem}
\end{align}
\normalsize
As $\sum_{n=m}^\infty \|\e^{-(n^2-2)\pi^2\alpha^2 t/(L_2-L_1)^2}\|_{\leb^1(\Omega)}<\infty$, we conclude that 
\[ \lim_{N\rightarrow\infty} f_{u_N(x,t)}(u)=f_{u(x,t)}(u) \]
in $\leb^\infty(\mathbb{R})$, with convergence rate given by (\ref{anem}).

\end{proof}

\begin{theorem} \label{superllei}
Let the random initial condition $\{\phi(x):\,L_1\leq x\leq L_2\}$ be a process in $\leb^2([L_1,L_2]\times \Omega)$. Let the random boundary conditions $A$ and $B$ belong to $\leb^2(\Omega)$. Suppose that $\alpha^2$, $A_1$, $(A_2,\ldots,A_N)$, $A$ and $B$ are independent and absolutely continuous, for $N\geq2$ (recall that $A_n$ is defined in (\ref{An}) as the random Fourier coefficient of $\psi$, where $\psi$ is defined from $\phi$ in relation (\ref{psi_w})). Suppose that the probability density function $f_{A_1}$ is a.e continuous on $\mathbb{R}$ and $\|f_{A_1}\|_{\leb^\infty(\mathbb{R})}<\infty$. Assume that $\mathbb{E}[\e^{\pi^2\alpha^2 t/(L_2-L_1)^2}]<\infty$. Then the sequence
\small
\[ f_{u_N(x,t)}(u)=\int_{\mathbb{R}}\int_{\mathbb{R}} f_{v_N\left(\frac{x-L_1}{L_2-L_1},t\right)}(u-b-a)f_B\left(\frac{L_2-L_1}{x-L_1}b\right)\frac{L_2-L_1}{x-L_1}f_A\left(\frac{L_2-L_1}{L_2-x}a\right)\frac{L_2-L_1}{L_2-x}\,\dif a\,\dif b, \]
\normalsize
where $f_{v_N}$ is the density defined by (\ref{fr}), converges pointwise to the density $f_{u(x,t)}(u)$ of the solution stochastic process $u(x,t)(\omega)$ to the randomized heat equation (\ref{edp_determinista_comp}), for $L_1<x<L_2$ and $t>0$.
\end{theorem}
\begin{proof}
From $\phi\in \leb^2([L_1,L_2]\times \Omega)$ and (\ref{psiL2}), it follows $\psi\in \leb^2([0,1]\times \Omega)$. By hypothesis, we also have that $\beta^2=\alpha^2/(L_2-L_1)^2$, $A_1$ and $(A_2,\ldots,A_N)$ are independent and absolutely continuous, for $N\geq2$, and $\mathbb{E}[\e^{\pi^2\beta^2 t}]<\infty$. Thereby, the hypotheses of Theorem \ref{teorllei} are fulfilled.

Since $\alpha^2$, $A_1$, $(A_2,\ldots,A_N)$, $A$ and $B$ are independent, as we did in the proof of Theorem \ref{super} we deduce that $v_N(y,t)$, $A$ and $B$ are independent, and that $v(y,t)$, $A$ and $B$ are independent. Hence,
\small
\[ f_{u_N(x,t)}(u)=\int_{\mathbb{R}}\int_{\mathbb{R}} f_{v_N\left(\frac{x-L_1}{L_2-L_1},t\right)}(u-b-a)f_B\left(\frac{L_2-L_1}{x-L_1}b\right)\frac{L_2-L_1}{x-L_1}f_A\left(\frac{L_2-L_1}{L_2-x}a\right)\frac{L_2-L_1}{L_2-x}\,\dif a\,\dif b \]
\normalsize
and
\small
\[ f_{u(x,t)}(u)=\int_{\mathbb{R}}\int_{\mathbb{R}} f_{v\left(\frac{x-L_1}{L_2-L_1},t\right)}(u-b-a)f_B\left(\frac{L_2-L_1}{x-L_1}b\right)\frac{L_2-L_1}{x-L_1}f_A\left(\frac{L_2-L_1}{L_2-x}a\right)\frac{L_2-L_1}{L_2-x}\,\dif a\,\dif b. \]
\normalsize

By Theorem \ref{teorllei},
\begin{align*}
 \lim_{N\rightarrow\infty} {} & f_{v_N\left(\frac{x-L_1}{L_2-L_1},t\right)}(u-b-a)f_B\left(\frac{L_2-L_1}{x-L_1}b\right)\frac{L_2-L_1}{x-L_1}f_A\left(\frac{L_2-L_1}{L_2-x}a\right)\frac{L_2-L_1}{L_2-x} \\
= {} & f_{v\left(\frac{x-L_1}{L_2-L_1},t\right)}(u-b-a)f_B\left(\frac{L_2-L_1}{x-L_1}b\right)\frac{L_2-L_1}{x-L_1}f_A\left(\frac{L_2-L_1}{L_2-x}a\right)\frac{L_2-L_1}{L_2-x},
\end{align*}
for every $u,a,b\in\mathbb{R}$, $L_1<x<L_2$ and $t>0$. By (\ref{dct}),
\small
\begin{align*}
{} & \left|f_{v_N\left(\frac{x-L_1}{L_2-L_1},t\right)}(u-b-a)f_B\left(\frac{L_2-L_1}{x-L_1}b\right)\frac{L_2-L_1}{x-L_1}f_A\left(\frac{L_2-L_1}{L_2-x}a\right)\frac{L_2-L_1}{L_2-x}\right| \\
\leq {} & \|f_{A_1}\|_{\leb^\infty(\mathbb{R})}\frac{\mathbb{E}[\e^{\pi^2\beta^2t}]}{\sin\left(\pi\frac{x-L_1}{L_2-L_1}\right)}\frac{(L_2-L_1)^2}{(x-L_1)(L_2-x)}f_B\left(\frac{L_2-L_1}{x-L_1}b\right)f_A\left(\frac{L_2-L_1}{L_2-x}a\right)\in\leb^1(\mathbb{R}^2,\dif a\,\dif b),
\end{align*}
\normalsize
so by the Dominated Convergence Theorem,
\[ \lim_{N\rightarrow \infty} f_{u_N(x,t)}(u)=f_{u(x,t)}(u) \]
follows.
\end{proof}

Theorem \ref{super} and Theorem \ref{superllei} may be adapted to the case in which $A$ and $B$ are deterministic. By applying Lemma \ref{lema_abscont} in (\ref{sol_u_w}) and (\ref{sol_u_w_N}), 
\[ f_{u(x,t)}(u)=f_{v\left(\frac{x-L_1}{L_2-L_1},t\right)}\left(u-\frac{x-L_1}{L_2-L_1}B-\frac{L_2-x}{L_2-L_1}A\right) \]
and
\begin{equation}
 f_{u_N(x,t)}(u)=f_{v_N\left(\frac{x-L_1}{L_2-L_1},t\right)}\left(u-\frac{x-L_1}{L_2-L_1}B-\frac{L_2-x}{L_2-L_1}A\right). 
 \label{fN2}
\end{equation}
One arrives at the following two theorems, which are proved similarly but easier than Theorem \ref{super} and Theorem \ref{superllei}, respectively.

\begin{theorem} \label{super_det}
Let the random initial condition $\{\phi(x):\,L_1\leq x\leq L_2\}$ be a process in $\leb^2([L_1,L_2]\times \Omega)$. Suppose that $\alpha^2$, $A_1$ and $(A_2,\ldots,A_N)$ are independent and absolutely continuous, for $N\geq2$ (recall that $A_n$ is defined in (\ref{An}) as the random Fourier coefficient of $\psi$, where $\psi$ is defined from $\phi$ in relation (\ref{psi_w})). Suppose that the probability density function $f_{A_1}$ is Lipschitz on $\mathbb{R}$ and that 
\[ \sum_{n=m}^\infty \|\e^{-(n^2-2)\pi^2\alpha^2 t/(L_2-L_1)^2}\|_{\leb^1(\Omega)}<\infty, \] 
for certain $m\in\mathbb{N}$. Then the sequence
\small
\[ f_{u_N(x,t)}(u)=f_{v_N\left(\frac{x-L_1}{L_2-L_1},t\right)}\left(u-\frac{x-L_1}{L_2-L_1}B-\frac{L_2-x}{L_2-L_1}A\right), \]
\normalsize
where $f_{v_N}$ is the density defined by (\ref{fr}), converges in $\leb^\infty(\mathbb{R})$ to the density $f_{u(x,t)}(u)$ of the solution stochastic process $u(x,t)(\omega)$ to the randomized heat equation (\ref{edp_determinista_comp}) with deterministic boundary conditions $A$ and $B$, for $L_1<x<L_2$ and $t>0$.
\end{theorem}

\begin{theorem} \label{super_det_llei}
Let the random initial condition $\{\phi(x):\,L_1\leq x\leq L_2\}$ be a process in $\leb^2([L_1,L_2]\times \Omega)$. Suppose that $\alpha^2$, $A_1$ and $(A_2,\ldots,A_N)$ are independent and absolutely continuous, for $N\geq2$ (recall that $A_n$ is defined in (\ref{An}) as the random Fourier coefficient of $\psi$, where $\psi$ is defined from $\phi$ in relation (\ref{psi_w})). Suppose that the probability density function $f_{A_1}$ is a.e. continuous on $\mathbb{R}$ and $\|f_{A_1}\|_{\leb^\infty(\mathbb{R})}<\infty$. Assume that $\mathbb{E}[\e^{\pi^2\alpha^2 t/(L_2-L_1)^2}]<\infty$. Then the sequence
\small
\[ f_{u_N(x,t)}(u)=f_{v_N\left(\frac{x-L_1}{L_2-L_1},t\right)}\left(u-\frac{x-L_1}{L_2-L_1}B-\frac{L_2-x}{L_2-L_1}A\right), \]
\normalsize
where $f_{v_N}$ is the density defined by (\ref{fr}), converges pointwise to the density $f_{u(x,t)}(u)$ of the solution stochastic process $u(x,t)(\omega)$ to the randomized heat equation (\ref{edp_determinista_comp}) with deterministic boundary conditions $A$ and $B$, for $L_1<x<L_2$ and $t>0$.
\end{theorem}

\section{Approximation of the expectation and variance of the solution stochastic process}

By Theorem \ref{st_sol} ii), if $\psi\in\leb^2([0,1]\times\Omega)$ (this holds if $\phi\in\leb^2([L_1,L_2]\times\Omega)$ and $A,B\in\leb^2(\Omega)$, by (\ref{psiL2})) and $0<a\leq\beta^2(\omega)\leq b$ a.s., then $v_N(y,t)\rightarrow v(y,t)$ in $\leb^2(\Omega)$ as $N\rightarrow\infty$. 

In fact, looking at the proof of \cite[Th. 1.3]{jjcm}, we can be more precise: in that proof, it was shown that $\|A_n\|_{\leb^2(\Omega)}\leq C$, for all $n$. If we assume that $\beta^2$ and $A_n$ are independent, for each $n$, and that $\sum_{n=1}^\infty \|\e^{-n^2\pi^2\beta^2 t}\|_{\leb^2(\Omega)}<\infty$, then
\begin{align*}
 \sum_{n=1}^\infty \|A_n \e^{-n^2\pi^2\beta^2 t}\sin(n\pi y)\|_{\leb^2(\Omega)}\leq {} & \sum_{n=1}^\infty \|A_n\|_{\leb^2(\Omega)}\|\e^{-n^2\pi^2\beta^2 t}\|_{\leb^2(\Omega)} \\
\leq {} & C\sum_{n=1}^\infty \|\e^{-n^2\pi^2\beta^2 t}\|_{\leb^2(\Omega)}<\infty, 
\end{align*}
which implies that $v_N(y,t)\rightarrow v(y,t)$ in $\leb^2(\Omega)$ as $N\rightarrow\infty$. By (\ref{sol_u_w}) and (\ref{sol_u_w_N}), this is equivalent to $u_N(x,t)\rightarrow u(x,t)$ in $\leb^2(\Omega)$ as $N\rightarrow\infty$.

We already know that, if $v_N(y,t)(\omega)$, $A(\omega)$ and $B(\omega)$ are absolutely continuous and independent random variables, then $u_N(x,t)(\omega)$ has a density function $f_{u_N(x,t)}(u)$ given by (\ref{fN}). On the other hand, if $A$ and $B$ are deterministic, assuming that $v_N(y,t)(\omega)$ is absolutely continuous one has that $u_N(x,t)(\omega)$ has a density function $f_{u_N(x,t)}(u)$ expressed by (\ref{fN2}). Thus,
\begin{equation}\label{media_aprox}
\mathbb{E}[u_N(x,t)]=\int_\mathbb{R} u\,f_{u_N(x,t)}(u)\,\dif u 
\end{equation}
and
\begin{equation}\label{varianza_aprox}
 \mathbb{V}[u_N(x,t)]=\int_\mathbb{R} u^2\,f_{u_N(x,t)}(u)\,\dif u-\left(\mathbb{E}[u_N(x,t)]\right)^2. 
\end{equation}

We summarize these ideas in the following theorem and remark, where the random or deterministic nature of the parameters $A$ and $B$ is distinguished, respectively, for the sake of completeness in the statement of our findings:

\begin{theorem} \label{e1}
If $\phi\in\leb^2([L_1,L_2]\times\Omega)$, $A,B\in\leb^2(\Omega)$, $\alpha^2,(A_1,\ldots,A_N),A,B$ are absolutely continuous and independent, and $\sum_{n=1}^\infty \|\e^{-n^2\pi^2\alpha^2 t/(L_2-L_1)^2}\|_{\leb^2(\Omega)}<\infty$, then $u(x,t)\in\leb^2(\Omega)$,
\[ \mathbb{E}[u_N(x,t)]=\int_\mathbb{R} u\,f_{u_N(x,t)}(u)\,\dif u\stackrel{N\rightarrow\infty}{\longrightarrow} \mathbb{E}[u(x,t)] \]
and
\[ \mathbb{V}[u_N(x,t)]=\int_\mathbb{R} u^2\,f_{u_N(x,t)}(u)\,\dif u-\left(\mathbb{E}[u_N(x,t)]\right)^2 \stackrel{N\rightarrow\infty}{\longrightarrow} \mathbb{V}[u(x,t)], \]
for each $L_1<x<L_2$ and $t>0$.
\end{theorem}

\begin{remark} \label{nota_A_B_deterministicos}
Theorem \ref{e1} holds in the case that $A$ and $B$ are deterministic values.
\end{remark}


\section{Applications}

The first question that arises is to which random diffusion coefficients, random boundary conditions and random initial conditions our results can be applied. 

We begin by studying hypothesis
\begin{equation}
 \sum_{n=m}^\infty \|\e^{-(n^2-2)\pi^2\alpha^2 t/(L_2-L_1)^2}\|_{\leb^1(\Omega)}<\infty, 
 \label{hip_a}
\end{equation}
for $t>0$. It is clear that if $\alpha^2$ is bounded below, meaning that $\alpha^2(\omega)\geq a>0$ for a.e. $\omega\in\Omega$, then (\ref{hip_a}) holds. This covers all cases in practice, as we may truncate $\alpha^2$, \cite{truncate}. Notice, however, that the condition $\alpha^2(\omega)\geq a>0$ is not necessary to have (\ref{hip_a}). For example, if $\alpha^2\sim\text{Uniform}(0,b)$, $b>0$, then we know that its moment generating function is given by
\begin{equation}
 \mathbb{E}[\e^{\lambda \alpha^2}]=\frac{\e^{\lambda b}-1}{\lambda b}, 
 \label{momentunif}
\end{equation}
therefore 
\begin{align*}
 \sum_{n=m}^\infty \|\e^{-(n^2-2)\pi^2\alpha^2 t/(L_2-L_1)^2}\|_{\leb^1(\Omega)}= {} & \sum_{n=m}^\infty \frac{\e^{-(n^2-2)\pi^2 b t/(L_2-L_1)^2}-1}{-(n^2-2)\pi^2 b t/(L_2-L_1)^2} \\
\leq {} & \sum_{n=m}^\infty\frac{1}{(n^2-2)\pi^2 b t /(L_2-L_1)^2}<\infty. 
\end{align*}
Another distribution for which (\ref{hip_a}) holds, this time not upper-bounded, is $\alpha^2\sim\text{Gamma}(r,s)$, being $r>1/2$ the shape and $s>0$ the rate. Its moment generating function is given by 
\[ \mathbb{E}[\e^{\lambda \alpha^2}]=\frac{1}{\left(1-\frac{\lambda}{s}\right)^r}, \]
for $\lambda<s$. Then
\begin{equation}
 \sum_{n=m}^\infty \|\e^{-(n^2-2)\pi^2\alpha^2 t/(L_2-L_1)^2}\|_{\leb^1(\Omega)}=\sum_{n=m}^\infty \frac{1}{[(1+(n^2-2)\pi^2 \alpha^2 t/(s(L_2-L_1)^2)]^r}<\infty. 
 \label{eq_gamma}
\end{equation}
Notice that, if $0<r\leq 1/2$, then $\sum_{n=m}^\infty \|\e^{-(n^2-2)\pi^2\alpha^2 t/(L_2-L_1)^2}\|_{\leb^1(\Omega)}=\infty$. This shows that hypothesis (\ref{hip_a}) might not hold.

Concerning hypothesis 
\begin{equation}
 \mathbb{E}[\e^{\pi^2\alpha^2 t/(L_2-L_1)^2}]<\infty,
 \label{hip_a2}
\end{equation}
for $t>0$, just take any distribution with finite moment generating function for $t>0$. For instance, $\text{Uniform}(0,b)$ with moment generating function (\ref{momentunif}), $\text{Normal}(\mu,\sigma^2)$ with moment generating function at $\lambda$ given by $\e^{\mu \lambda+1/2\sigma^2 \lambda^2}$, etc. 

The gamma distribution may be used to highlight the fact that, fixed $t>0$, hypotheses (\ref{hip_a}) and (\ref{hip_a2}) are independent. Suppose that $\alpha^2\sim\text{Gamma}(r,s)$, being $r>0$ the shape and $s>0$ the rate. Then (\ref{hip_a}) is accomplished if and only if $r>1/2$ and $s>0$ (see (\ref{eq_gamma})), whereas (\ref{hip_a2}) fulfills if and only if $r>0$ and $\pi^2 t/(L_2-L_1)^2<s$. 

The most difficult step is to compute $f_{A_1}$ and $f_{(A_2,\ldots,A_N)}$ in (\ref{fr}). We are going to see that the density function of
\[ A_n(\omega)=2\int_0^1 \psi(y)(\omega)\sin(n\pi y)\,\dif y\]
can be computed when the initial condition process $\phi$ has a certain expression concerning the Karhunen-Loève expansion. Take $\psi$ defined in (\ref{psi_w}). As $\psi\in \leb^2([0,1]\times \Omega)$, for each fixed $\omega\in\Omega$ the real function $\psi(\cdot)(\omega)$ belongs to $\leb^2([0,1])$. We can expand $\psi(\cdot)(\omega)$ as a Fourier series on $[0,1]$ with the orthonormal basis $\{\sqrt{2}\sin(j\pi y)\}_{j=1}^\infty$. Hence,
\begin{equation}
 \psi(y)(\omega)=\sum_{j=1}^\infty c_j(\omega)\sqrt{2}\sin(j\pi y), 
 \label{fouri}
\end{equation}
where the series is taken in $\leb^2([0,1])$ for each $\omega\in\Omega$, and where $c_j(\omega)$ are the random variables corresponding to the Fourier coefficients of $\psi(\cdot)(\omega)$. This expression (\ref{fouri}) corresponds to the Karhunen-Loève expansion of the process $\psi$. We will restrict to processes for which the random Fourier coefficients $\{c_j\}_{j=1}^\infty$ are independent and absolutely continuous random variables. Thus, in the notation of Lemma \ref{KLlemma}, we write
\begin{equation}
 \psi(y)(\omega)=\sum_{j=1}^\infty \sqrt{\nu_j}\sqrt{2}\sin(j\pi y)\xi_j(\omega), 
 \label{KLpsi}
\end{equation}
where the series converges in $\leb^2([0,1]\times \Omega)$, $\{\nu_j\}_{j=1}^\infty$ are nonnegative real numbers satisfying $\sum_{j=1}^\infty \nu_j<\infty$ and $\{\xi_j\}_{j=1}^\infty$ are absolutely continuous random variables with zero expectation, unit variance and independent ($c_j(\omega)=\sqrt{\nu_j}\xi_j(\omega)$, so that $\xi_j$ standardizes $c_j(\omega)$). Notice that the sum is well-defined in $\leb^2([0,1]\times\Omega)$, because for two indexes $N>M$ we have, by Pythagoras's Theorem in $\leb^2([0,1]\times\Omega)$,
\begin{align}
 \left\|\sum_{j=M+1}^N \sqrt{\nu_j} \,\sqrt{2}\,\sin(j\pi x)\,\xi_j\right\|_{\leb^2([0,1]\times\Omega)}^2= {} & \sum_{j=M+1}^N \nu_j\, \|\sqrt{2}\,\sin(j\pi x)\|_{\leb^2([0,1])}^2\|\xi_j\|_{\leb^2(\Omega)}^2 \nonumber \\
= {} & \sum_{j=M+1}^N \nu_j\stackrel{N,M\rightarrow\infty}{\longrightarrow}0. \label{pyth}
\end{align}

We can compute explicitly the random Fourier coefficients $A_n$:
\begin{align}
 A_n(\omega)= {} & 2\int_0^1 \psi(y)(\omega)\sin(n\pi y)\,\dif y=2\sum_{j=1}^\infty \sqrt{\nu_j}\,\sqrt{2}\,\int_0^1 \sin(j\pi y)\sin(n\pi y)\,\dif y\,\xi_j(\omega) \nonumber \\
= {} & \sqrt{2}\,\sqrt{\nu_n}\,\xi_n(\omega). \label{A1x1}
\end{align}
The key fact in this computation is that the eigenfunctions of the Sturm-Liouville problem associated to (\ref{edp_determinista_homo}) are precisely $\{\sqrt{2}\sin(j\pi y)\}_{j=1}^\infty$. From (\ref{A1x1}) and our assumptions on $\{\xi_j\}_{j=1}^\infty$, we derive that $A_1,A_2,\ldots$ are independent and absolutely continuous  random variables. Using Lemma \ref{lema_abscont},
\[ f_{A_n}(a)=\frac{1}{\sqrt{2\nu_n}} f_{\xi_n}\left(\frac{a}{\sqrt{2\nu_n}}\right). \]
If $f_{\xi_1}$ is Lipschitz (respectively a.e. continuous and essentially bounded) on $\mathbb{R}$, then $f_{A_1}$ is Lipschitz (respectively a.e. continuous and essentially bounded) on $\mathbb{R}$ too, and all the hypotheses of Theorem \ref{super} (respectively Theorem \ref{superllei}) are fulfilled.

The Lipschitz condition on $\mathbb{R}$ is satisfied by the probability density function of some named distributions:
\begin{itemize}
\item $\text{Normal}(\mu,\sigma^2)$, $\mu\in\mathbb{R}$ and $\sigma^2>0$.
\item $\text{Beta}(a,b)$, $a,b\geq2$.
\item $\text{Gamma}(a,b)$, $a\geq2$ and $b>0$.
\end{itemize}
In general, any density with bounded derivative on $\mathbb{R}$ satisfies the Lipschitz condition on $\mathbb{R}$, by the Mean Value Theorem. By contrast, some non-Lipschitz density functions are the uniform distribution, the exponential distribution, etc. or any other density with a jump discontinuity at some point of $\mathbb{R}$. However, non-Lipschitz density functions may be regularized at the point of discontinuity so that the Lipschitz assumption is fulfilled and, moreover, the probabilistic behavior of the regularized density function is the same in practice as the original non-Lipschitz density.

The a.e. continuity and essential boundedness is satisfied by the probability density function of more distributions:
\begin{itemize}
\item $\text{Normal}(\mu,\sigma^2)$, $\mu\in\mathbb{R}$ and $\sigma^2>0$.
\item $\text{Beta}(a,b)$, $a,b\geq1$. 
\item $\text{Uniform}(a,b)$, $a<b$.
\item $\text{Gamma}(a,b)$, $a\geq1$ and $b>0$. In particular, $\text{Exponential}(\lambda)$, $\lambda>0$.
\item Truncated normal distribution.
\end{itemize}

We will do examples for initial conditions $\phi$ such that the corresponding $\psi$ is written as (\ref{KLpsi}), being $\xi_1,\xi_2,\ldots$ independent and absolutely continuous random variables, with zero expectation and unit variance, $f_{\xi_1}$ Lipschitz on $\mathbb{R}$ and $\sum_{j=1}^\infty \nu_j<\infty$. In the examples, we will combine $A$ and $B$ deterministic and absolutely continuous random variables, with $\psi$ being a Gaussian and non-Gaussian process. Hence, all the examples suppose an improvement of \cite{jjcm}.

The densities $f_{u_N(x,t)}(u)$ that approximate $f_{u(x,t)}(u)$ will be computed numerically in an (almost) exact manner, using the software Mathematica\textsuperscript{\tiny\textregistered}, concretely, its built-in function \verb|NIntegrate|. In this way, we will be able to study the exact difference between two consecutive orders of truncation $N$ and $N+1$. 

\begin{example}[The process $\psi$ is Gaussian, the boundary conditions $A$ and $B$ are deterministic]	\label{ex1} \normalfont

Let 
\[ \psi(y)(\omega)=\sum_{j=1}^\infty \frac{\sqrt{2}}{\pi j}\sin(j\pi y)\xi_j(\omega) \]
be a standard Brownian bridge on $[0,1]$, see \cite[Example 5.30]{llibre_powell}, being $\xi_1,\xi_2,\ldots$ independent and $\text{Normal}(0,1)$ random variables. By (\ref{psi_w}), $\phi$ is a Brownian bridge on $[L_1,L_2]$ that takes on the values $A$ and $B$ at the boundary. We choose $L_1=0$ and $L_2=6$, $A=-3$ and $B=3$. The diffusion coefficient is $\alpha^2\sim\text{Uniform}(1,2)$. Theorem \ref{super_det} applies in this case. 

In Figure \ref{ex1phi}, three plots of the path described by $\phi(x)$ for three different outcomes $\omega$ are shown.

\begin{figure}[hbtp!]
  \begin{center}
    \includegraphics[width=4.9cm]{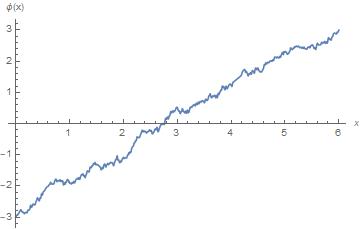}
		\includegraphics[width=4.9cm]{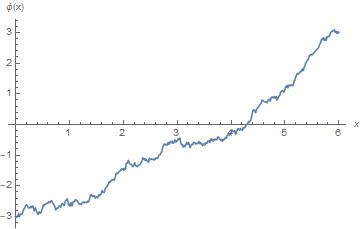}
		\includegraphics[width=4.9cm]{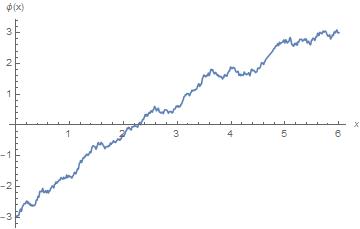}
    \caption{Paths of the initial condition $\phi(x)$ for three different outcomes $\omega$. Example \ref{ex1}.}
		\label{ex1phi}
    \end{center}
  \end{figure}
	
In Figure \ref{ex1dens}, we approximate the probability density function of the solution stochastic process $u(x,t)(\omega)$ at $x=5$ and $t=0.2$, using (\ref{fN2}), for $N=1,2,3,4$. Convergence seems to be achieved. In Figure \ref{ex1dens3D}, a three-dimensional plot of (\ref{fN2}) for $x=5$, $N=2$ and $t$ varying is shown. In Table \ref{ex1densTable}, the infinity norm of the difference of two consecutive orders of approximation $N$ and $N+1$, for $N=1,2,3$, is computed. We can see that the errors decrease to $0$ as $N$ grows, which agrees with our theoretical findings. In Table \ref{ex1densE}, using Theorem \ref{e1} together with Remark \ref{nota_A_B_deterministicos}, the expectation and variance of $u(x,t)(\omega)$ have been approximated, for different orders of truncation. 

\begin{figure}[hbtp!]
  \begin{center}
    \includegraphics[width=6cm]{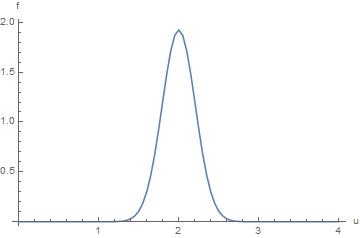}
		\includegraphics[width=6cm]{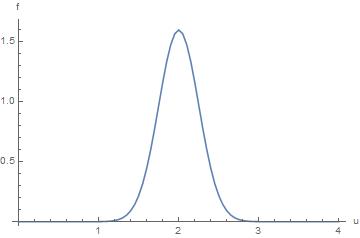}
		\includegraphics[width=6cm]{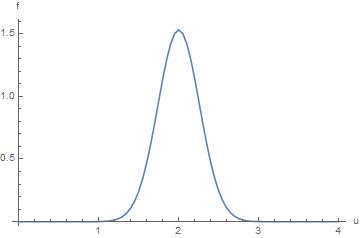}
		\includegraphics[width=6cm]{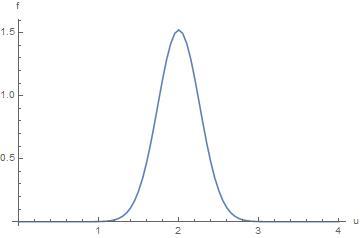}
    \caption{Approximation (\ref{fN2}) for $N=1$ (up left), $N=2$ (up right), $N=3$ (down left) and $N=4$ (down right), at $x=5$ and $t=0.2$. Example \ref{ex1}.}
		\label{ex1dens}
    \end{center}
  \end{figure}
	
\begin{figure}[hbtp!]
  \begin{center}
    \includegraphics[width=5cm]{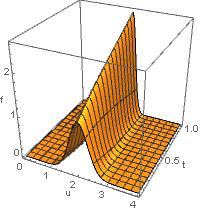}
    \caption{Approximation (\ref{fN2}) for $N=2$ at $x=5$ and $t$ varying. Example \ref{ex1}.}
		\label{ex1dens3D}
    \end{center}
  \end{figure}
	
\begin{table}[hbtp!]
\begin{center}
\begin{tabular}{|c|c|} \hline
$N$ & $\|f_{u_N(5,0.2)}-f_{u_{N+1}(5,0.2)}\|_{\leb^\infty(\mathbb{R})}$ \\ \hline
$1$ & $0.330855$ \\ \hline
$2$ & $0.0622449$ \\ \hline
$3$ & $0.00820879$ \\ \hline
\end{tabular}
\caption{Infinity norm of the difference of two consecutive orders of approximation $N$ and $N+1$ given by (\ref{fN2}), for $N=1,2,3$. Example \ref{ex1}.}
\label{ex1densTable}
\end{center}
\end{table}

\begin{table}[hbtp!]
\begin{center}
\begin{tabular}{|c|c|c|c|c|} \hline
$N$ & $1$ & $2$ & $3$ & $4$ \\ \hline
$\mathbb{E}[u_N(5,0.2)]$ & $2$ & $2$ & $2$ & $2$  \\ \hline
$\mathbb{V}[u_N(5,0.2)]$ & $0.0429981$ & $0.0628341$ & $0.0681679$ & $0.0689422$ \\ \hline
\end{tabular}
\caption{Approximations of $\mathbb{E}[u(5,0.2)]$ and $\mathbb{V}[u(5,0.2)]$ for $N=1,2,3,4$ constructed by (\ref{media_aprox}) and (\ref{varianza_aprox}), respectively, being $f_{u_N(x,t)}$ given by  (\ref{fN2}). Example \ref{ex1}.}
\label{ex1densE}
\end{center}
\end{table}

\end{example}

\begin{example}[The process $\psi$ is non-Gaussian, the boundary conditions $A$ and $B$ are deterministic]	\label{ex2} \normalfont

Let 
\[ \psi(y)(\omega)=\sum_{j=1}^\infty \frac{\sqrt{2}}{j^{\frac32}\sqrt{1+\log j}}\sin(j\pi y)\xi_j(\omega), \]
where $\xi_1,\xi_2,\ldots$ are independent and identically distributed random variables with density function
\[ f_{\xi_1}(\xi)=\frac{\sqrt{2}}{\pi(1+\xi^4)},\quad -\infty < \xi < \infty. \]
It is easy to check that this is indeed a density function, with zero expectation and unit variance. Thereby, $\psi$ is a non-Gaussian stochastic process on $[0,1]$ (if it were Gaussian, by the last assertion in the statement of Lemma \ref{KLlemma}, $\xi_1,\xi_2,\ldots$ would be normally distributed). The sum defining $\psi$ is well-defined in $\leb^2([0,1]\times\Omega)$, because $\sum_{j=1}^\infty 1/(j^3(1+\log j))<\infty$ (see (\ref{pyth})). By (\ref{psi_w}), we can simulate the sample paths of $\phi(x)$ on $[L_1,L_2]$. The data chosen are $L_1=-8$, $L_2=2\pi+1$, $A=-1$ and $B=2$. The distribution for $\alpha^2$ is $\text{Uniform}(1,2)$. Theorem \ref{super_det} guarantees the convergence of the approximating sequence (\ref{fN2}). 

In Figure \ref{ex2phi}, three plots of the path described by $\phi(x)$ for three different outcomes $\omega$ are presented.

\begin{figure}[hbtp!]
  \begin{center}
    \includegraphics[width=4.9cm]{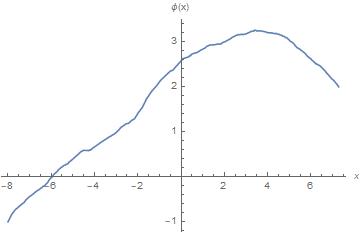}
		\includegraphics[width=4.9cm]{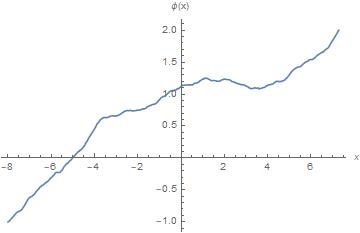}
		\includegraphics[width=4.9cm]{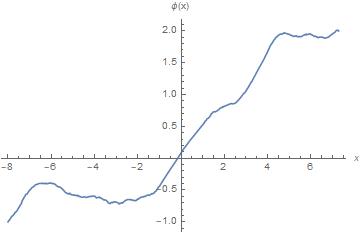}
    \caption{Paths of the initial condition $\phi(x)$ for three different outcomes $\omega$. Example \ref{ex2}.}
		\label{ex2phi}
    \end{center}
  \end{figure}
	
In Figure \ref{ex2dens}, we approximate the probability density function of the solution stochastic process $u(x,t)(\omega)$ at $x=1$ and $t=0.1$, using (\ref{fN2}), for $N=1,2,3,4$. In Figure \ref{ex2dens3D}, a three-dimensional plot of (\ref{fN2}) for $x=1$, $N=2$ and $t$ varying is shown, to see the evolution of the density as time goes on. In order to assess convergence analytically, in Table \ref{ex2densTable}, the maximum of the difference of two consecutive orders of approximation $N$ and $N+1$ given by (\ref{fN2}), for $N=1,2,3$, is computed. The errors decrease to $0$ as $N$ grows, which goes in the direction of our theoretical results. In Table \ref{ex2densE}, the expectation and variance of $u(x,t)(\omega)$ have been approximated, using Theorem \ref{e1} and Remark \ref{nota_A_B_deterministicos} together with expressions \eqref{media_aprox} and (\ref{varianza_aprox}), respectively.

\begin{figure}[hbtp!]
  \begin{center}
    \includegraphics[width=6cm]{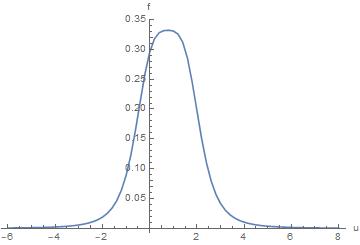}
		\includegraphics[width=6cm]{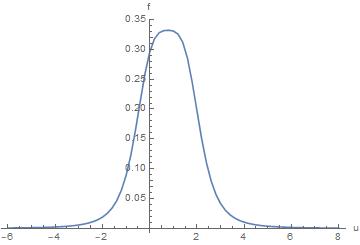}
		\includegraphics[width=6cm]{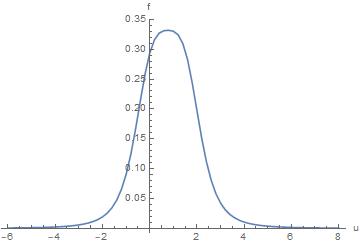}
		\includegraphics[width=6cm]{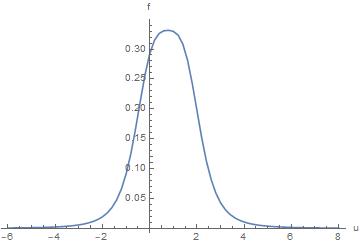}
    \caption{Approximation (\ref{fN2}) for $N=1$ (up left), $N=2$ (up right), $N=3$ (down left) and $N=4$ (down right), at $x=1$ and $t=0.1$. Example \ref{ex2}.}
		\label{ex2dens}
    \end{center}
  \end{figure}
	
\begin{figure}[hbtp!]
  \begin{center}
    \includegraphics[width=5cm]{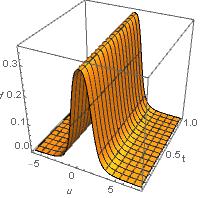}
    \caption{Approximation (\ref{fN2}) for $N=2$ at $x=1$ and $t$ varying. Example \ref{ex2}.}
		\label{ex2dens3D}
    \end{center}
  \end{figure}
	
\begin{table}[hbtp!]
\begin{center}
\begin{tabular}{|c|c|} \hline
$N$ & $\|f_{u_N(1,0.1)}-f_{u_{N+1}(1,0.1)}\|_{\leb^\infty(\mathbb{R})}$ \\ \hline
$1$ & $0.00631990$ \\ \hline
$2$ & $0.00214919$ \\ \hline
$3$ & $0.00128318$ \\ \hline 
\end{tabular}
\caption{Infinity norm of the difference of two consecutive orders of approximation $N$ and $N+1$ given by (\ref{fN2}), for $N=1,2,3$. Example \ref{ex2}.}
\label{ex2densTable}
\end{center}
\end{table}

\begin{table}[hbtp!]
\begin{center}
\begin{tabular}{|c|c|c|c|c|} \hline
$N$ & $1$ & $2$ & $3$ & $4$ \\ \hline
$\mathbb{E}[u_N(1,0.1)]$ & $1.54545$ & $1.54497$ & $1.54495$ & $1.54494$  \\ \hline
$\mathbb{V}[u_N(1,0.1)]$ & $1.51182$ & $1.55304$ & $1.56661$ & $1.57496$ \\ \hline
\end{tabular}
\caption{Approximation of $\mathbb{E}[u(1,0.1)]$ and $\mathbb{V}[u(1,0.1)]$ for $N=1,2,3,4$ constructed by (\ref{media_aprox}) and (\ref{varianza_aprox}), respectively, being $f_{u_N(x,t)}$ given by  (\ref{fN2}). Example \ref{ex2}.}
\label{ex2densE}
\end{center}
\end{table}

\end{example}

\begin{example}[The process $\psi$ is Gaussian, the boundary conditions $A$ and $B$ are random]	\label{ex3} \normalfont

Let 
\[ \psi(y)(\omega)=\sum_{j=1}^\infty \frac{\sqrt{2}}{\pi j}\sin(j\pi y)\xi_j(\omega) \]
be a standard Brownian bridge on $[0,1]$, as in Example \ref{ex1}. The data chosen are $L_1=0$, $L_2=6$ and $\alpha^2\sim\text{Uniform}(1,2)$, as in Example \ref{ex1}, but now the boundary conditions $A$ and $B$ are random: $A$ follows a triangular distribution with ends $-5$ and $-2$ and mode $-3$, whereas $B$ is an exponentially distributed random variable with mean $2$ and truncated to $[3,5]$. The modes of $A$ and $B$ coincide with the deterministic boundary conditions in Example \ref{ex1}, so similar results for the density function could occur.

In Figure \ref{ex3phi}, three plots of the path described by $\phi(x)$ for three different outcomes $\omega$ are presented.

\begin{figure}[hbtp!]
  \begin{center}
    \includegraphics[width=4.9cm]{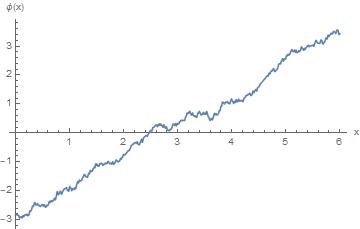}
		\includegraphics[width=4.9cm]{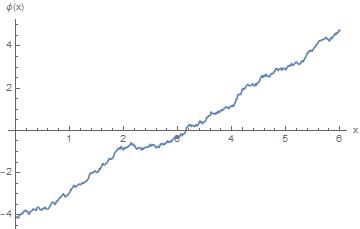}
		\includegraphics[width=4.9cm]{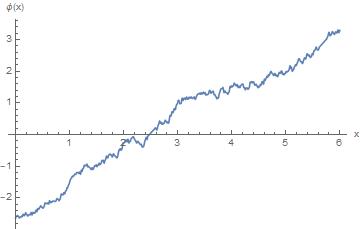}
    \caption{Paths of the initial condition $\phi(x)$ for three different outcomes $\omega$. Example \ref{ex3}.}
		\label{ex3phi}
    \end{center}
  \end{figure}
	
In Figure \ref{ex3dens}, we approximate the probability density function of the solution stochastic process $u(x,t)(\omega)$ at $x=5$ and $t=0.2$, using (\ref{fN}), for $N=1,2,3,4$. Compare the plots with those of Example \ref{ex1}, where the boundary conditions were deterministic with constant value the mode of $A$ and $B$. In Table \ref{ex3densTable}, the errors are analyzed. In Table \ref{ex3densE}, both $\mathbb{E}[u(x,t)]$ and $\mathbb{V}[u(x,t)]$ are approximated, according to Theorem \ref{e1}. 

\begin{figure}[hbtp!]
  \begin{center}
    \includegraphics[width=6cm]{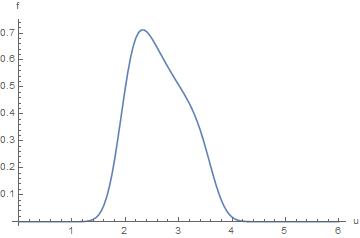}
		\includegraphics[width=6cm]{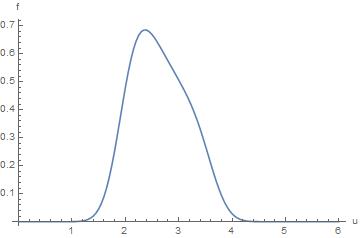}
		\includegraphics[width=6cm]{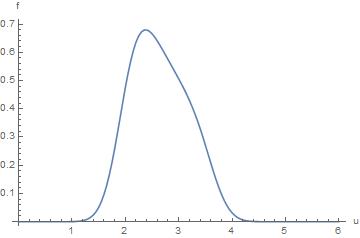}
		\includegraphics[width=6cm]{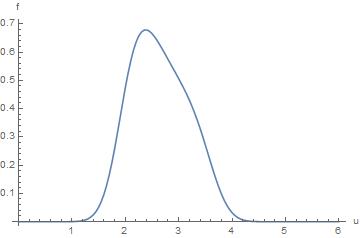}
    \caption{Approximation (\ref{fN}) for $N=1$ (up left), $N=2$ (up right), $N=3$ (down left) and $N=4$ (down right), at $x=5$ and $t=0.2$. Example \ref{ex3}.}
		\label{ex3dens}
    \end{center}
  \end{figure}
	
\begin{table}[hbtp!]
\begin{center}
\begin{tabular}{|c|c|} \hline
$N$ & $\|f_{u_N(5,0.2)}-f_{u_{N+1}(5,0.2)}\|_{\leb^\infty(\mathbb{R})}$ \\ \hline
$1$ & $0.0390501$ \\ \hline
$2$ & $0.00870084$ \\ \hline
$3$ & $0.00119532$ \\ \hline 
\end{tabular}
\caption{Infinity norm of the difference of two consecutive orders of approximation $N$ and $N+1$ given by (\ref{fN}), for $N=1,2,3$. Example \ref{ex3}.}
\label{ex3densTable}
\end{center}
\end{table}

\begin{table}[hbtp!]
\begin{center}
\begin{tabular}{|c|c|c|c|c|} \hline
$N$ & $1$ & $2$ & $3$ & $4$ \\ \hline
$\mathbb{E}[u_N(5,0.2)]$ & $2.64115$ & $2.64115$ & $2.64115$ & $2.64115$  \\ \hline
$\mathbb{V}[u_N(5,0.2)]$ & $0.274152$ & $0.293988$ & $0.299323$ & $0.300101$ \\ \hline
\end{tabular}
\caption{Approximation of $\mathbb{E}[u(5,0.2)]$ and $\mathbb{V}[u(5,0.2)]$ for $N=1,2,3,4$ constructed by (\ref{media_aprox}) and (\ref{varianza_aprox}), respectively, being $f_{u_N(x,t)}$ given by  (\ref{fN}). Example \ref{ex3}.}
\label{ex3densE}
\end{center}
\end{table}

\end{example}

\begin{example}[The process $\psi$ is non-Gaussian, the boundary conditions $A$ and $B$ are random]	\label{ex4} \normalfont

Let 
\[ \psi(y)(\omega)=\sum_{j=1}^\infty \frac{\sqrt{2}}{j^{\frac32}\sqrt{1+\log j}}\sin(j\pi y)\xi_j(\omega), \]
be the same process as in Example \ref{ex2}. The interval where the heat equation is defined has endpoints $L_1=-8$ and $L_2=2\pi+1$, and $\alpha^2\sim\text{Uniform}(1,2)$, as in Example \ref{ex2}. But now the boundary conditions $A$ and $B$ are random: $A\sim\text{Uniform}(-1.5,-0.5)$ and $B\sim \text{Normal}(2,1)$. Notice that $\mathbb{E}[A]$ and $\mathbb{E}[B]$ are the deterministic boundary conditions of Example \ref{ex2}, so the approximated density functions may resemble those from Example \ref{ex2}.

In Figure \ref{ex4phi}, three plots of the path described by $\phi(x)$ for three different outcomes $\omega$ are presented.

\begin{figure}[hbtp!]
  \begin{center}
    \includegraphics[width=4.9cm]{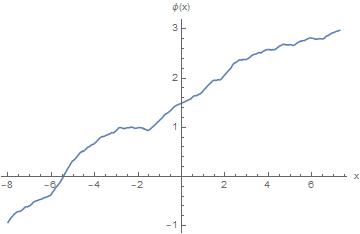}
		\includegraphics[width=4.9cm]{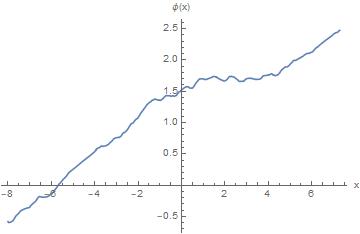}
		\includegraphics[width=4.9cm]{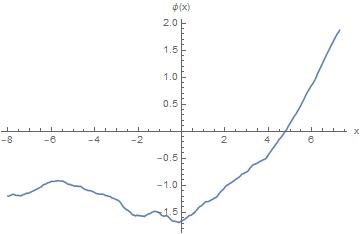}
    \caption{Paths of the initial condition $\phi(x)$ for three different outcomes $\omega$. Example \ref{ex4}.}
		\label{ex4phi}
    \end{center}
  \end{figure}
	
In Figure \ref{ex4dens}, we approximate the probability density function of the solution stochastic process $u(x,t)(\omega)$ at $x=1$ and $t=0.1$, using (\ref{fN}), for $N=1,2,3,4$. These plots are very similar to those from Example \ref{ex2}. This occurs because the expectation of our random boundary conditions $A$ and $B$ is equal to the deterministic boundary conditions of Example \ref{ex2}. In Table \ref{ex4densTable}, we present the errors between two consecutive orders of approximation. The expectation and variance of the solution process $u(x,t)(\omega)$ have been approximated in Table \ref{ex4densE}, based on Theorem \ref{e1}.

\begin{figure}[hbtp!]
  \begin{center}
    \includegraphics[width=6cm]{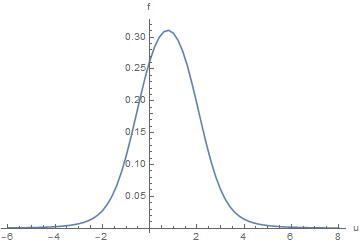}
		\includegraphics[width=6cm]{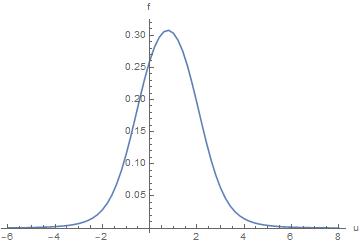}
		\includegraphics[width=6cm]{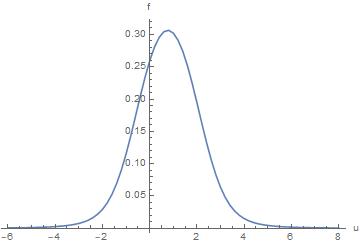}
		\includegraphics[width=6cm]{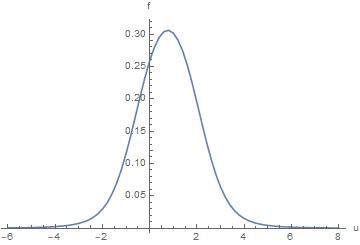}
    \caption{Approximation (\ref{fN}) for $N=1$ (up left), $N=2$ (up right), $N=3$ (down left) and $N=4$ (down right), at $x=1$ and $t=0.1$. Example \ref{ex4}.}
		\label{ex4dens}
    \end{center}
  \end{figure}
	
\begin{table}[hbtp!]
\begin{center}
\begin{tabular}{|c|c|} \hline
$N$ & $\|f_{u_N(1,0.1)}-f_{u_{N+1}(1,0.1)}\|_{\leb^\infty(\mathbb{R})}$ \\ \hline
$1$ & $0.00303628$ \\ \hline
$2$ & $0.00111420$ \\ \hline
$3$ & $0.000685738$ \\ \hline 
\end{tabular}
\caption{Infinity norm of the difference of two consecutive orders of approximation $N$ and $N+1$ given by (\ref{fN}), for $N=1,2,3$. Example \ref{ex4}.}
\label{ex4densTable}
\end{center}
\end{table}

\begin{table}[hbtp!]
\begin{center}
\begin{tabular}{|c|c|c|c|c|} \hline
$N$ & $1$ & $2$ & $3$ & $4$ \\ \hline
$\mathbb{E}[u_N(1,0.1)]$ & $0.766541$ & $0.766546$ & $0.766548$ & $0.766545$  \\ \hline
$\mathbb{V}[u_N(1,0.1)]$ & $1.86628$ & $1.90366$ & $1.91720$ & $1.92552$ \\ \hline
\end{tabular}
\caption{Approximation of $\mathbb{E}[u(1,0.1)]$ and $\mathbb{V}[u(1,0.1)]$ for $N=1,2,3,4$ constructed by (\ref{media_aprox}) and (\ref{varianza_aprox}), respectively, being $f_{u_N(x,t)}$ given by  (\ref{fN}). Example \ref{ex4}.}
\label{ex4densE}
\end{center}
\end{table}

\end{example}

\section{Conclusions}

In this paper we have determined approximations of the probability density function of the solution stochastic process to the randomized heat equation defined on a general bounded interval $[L_1,L_2]$ with non-homogeneous boundary conditions. We have reviewed results in the extant literature that establish conditions under which the probability density of the solution process to the random heat equation defined on $[0,1]$ with homogeneous boundary conditions can be approximated. By relating the solutions of the heat equation with homogeneous and non-homogeneous boundary conditions, and using the Random Variable Transformation technique, we have been able to set hypotheses on the random diffusion coefficient, on the random boundary conditions and on the initial condition process, so that the probability density function of the solution can be approximated uniformly or pointwise (Theorem \ref{super}, Theorem \ref{superllei}, Theorem \ref{super_det} and Theorem \ref{super_det_llei}). We have obtained results on the approximation of the expectation and variance of the solution (Theorem \ref{e1} and Remark \ref{nota_A_B_deterministicos}). 

Our theoretical findings have been applied to particular random heat equation problems on $[L_1,L_2]$ with non-homogeneous boundary conditions. We have dealt with random diffusion coefficients, with deterministic and random boundary conditions, and with initial condition processes having a certain Karhunen-Loève expansion, which may be Gaussian or may not. It has been evinced numerically that the convergence to the density function of the solution is achieved quickly.

\section*{Acknowledgements}
This work has been supported by Spanish Ministerio de Econom\'{i}a y Competitividad grant MTM2017--89664--P. Marc Jornet acknowledges the doctorate scholarship granted by Programa de Ayudas de Investigaci\'on y Desarrollo (PAID), Universitat Polit\`ecnica de Val\`encia.

\section*{Conflict of Interest Statement} 
The authors declare that there is no conflict of interests regarding the publication of this article.


\begin{thebibliography}{a}

\bibitem{llibre_smith}
\textsc{Ralph C. Smith}. \textit{Uncertainty Quantification. Theory, Implementation and Applications}. SIAM Computational Science \& Engineering, SIAM, Philadelphia, 2014.

\bibitem{Oksendal_solo}
\textsc{Bernt {\O}ksendal}. \textit{Stochastic Differential Equations. An Introduction with Applications}. Springer-Verlag, Series: Stochastic Modelling and Applied Probability 23, Heidelberg and New York, 2003.

\bibitem{Oksendal_otros}
\textsc{Helge Holden}, \textsc{Bernt {\O}ksendal}, \textsc{Jan Uboe} and \textsc{Tusheng Zhang}. \textit{Stochastic Partial Differential Equations: A Modeling, White Noise Functional Approach}. 2nd Ed.,  Springer Science+Business Media LLC, Series:  Probability and Its Applications, New York, 2010.

\bibitem{Kloeden_Platen}
\textsc{Peter Kloeden} and \textsc{Eckhard Platen}. \textit{Numerical Solution of Stochastic Differential Equations}. Springer-Verlag, Berlin and Heidelberg, 2011.

\bibitem{Xu_APM_2014}
\textsc{Zhijie Xu}. \textit{A stochastic analysis of steady and transient heat conduction in random media using a homogenization approach}. Applied Mathematical Modelling 38(13) (2014) 3233--3243. doi:10.1016/j.apm.2013.11.044.

\bibitem{Chiba_APM_2009} 
\textsc{Ryoichi Chiba}. \textit{Stochastic heat conduction analysis of a functionally graded annular disc with spatially random heat transfer coefficients}. Applied Mathematical Modelling 33(1) (2009) 507--523. doi:10.1016/j.apm.2007.11.014.

\bibitem{Soong}
\textsc{T.T. Soong}. \textit{Random Differential Equations in Science and Engineering}. Academic Press, New York, 1973.  

\bibitem{CCCJ_APM_2014}
\textsc{M. C. Casab\'{a}n}, \textsc{R. Company}, \textsc{J. C. Cortés} and \textsc{L. Jódar}. \textit{Solving the random diffusion model in an infinite medium: A mean square approach}. Applied Mathematical Modelling 38(24) (2014) 5922--5933. doi:10.1016/j.apm.2014.04.063.

\bibitem{CCJ_JCAM_2016}
\textsc{M. C. Casab\'{a}n}, \textsc{J. C. Cortés} and \textsc{L. Jódar}. \textit{Solving random mixed heat problems: A random integral transform approach}. Journal of Computational and Applied Mathematics 291 (2016) 5--19. doi:10.1016/j.cam.2014.09.021.

\bibitem{Jin_Xiu_Zhu_JCP_2015}
\textsc{Shi Jin}, \textsc{Dongbin Xiu} and \textsc{Xueyu Zhu}. \textit{Asymptotic-preserving methods for hyperbolic and transport equations with random inputs and diffusive scalings}. Journal of Computational Physics 289  (2015) 35--52. doi:10.1016/j.jcp.2015.02.023.

\bibitem{Jin_Lu_JCP_2017}
\textsc{Shi Jin} and \textsc{Hanqing Lu}. \textit{An asymptotic-preserving stochastic Galerkin method for the radiative heat transfer equations with random inputs and diffusive scalings}. Journal of Computational Physics 334  (2017) 182--206. doi:10.1016/j.jcp.2016.12.033.

\bibitem{Dorini_JCP}
\textsc{F.A. Dorini} and \textsc{M. Cristina C. Cunha}. \textit{Statistical moments of the random linear transport equation}. Journal of Computational Physics 227(19) (2008) 8541--8550. doi:10.1016/j.jcp.2008.06.002.

\bibitem{Selim_AMC_2011}
\textsc{A. Hussein} and \textsc{M.M. Selim}. \textit{Solution of the stochastic radiative transfer equation with Rayleigh scattering using RVT technique}. Applied Mathematics and Computation 218(13) (2012) 7193--7203. doi:10.1016/j.amc.2011.12.088.

\bibitem{Selim_EPJP_2015}
\textsc{A. Hussein} and \textsc{M.M. Selim}. \textit{Solution of the stochastic generalized shallow-water wave equation using RVT technique}. European Physical Journal Plus (2015) 130:249. doi:10.1140/epjp/i2015-15249-3.

\bibitem{Xu_AMM_2016}
\textsc{Zhijie Xu}, \textsc{Ramakrishna Tipireddy} and \textsc{Guang Lin}. \textit{Analytical approximation and numerical studies of one-dimensional elliptic equation with random coefficients}. Applied Mathematical Modelling  40(9-10) (2016) 5542-5559. doi:10.1016/j.apm.2015.12.041.


\bibitem{jjcm}
\textsc{J. Calatayud}, \textsc{J.-C. Cortés} and \textsc{M. Jornet}. \textit{On the approximation of the probability density function of the randomized heat equation}. https://arxiv.org/pdf/1802.04190.pdf.

\bibitem{EDO_lineal}
\textsc{J. Calatayud}, \textsc{J.-C. Cortés} and \textsc{M. Jornet}. \textit{On the approximation of the probability density function of the randomized non-autonomous complete linear differential equation}. https://arxiv.org/pdf/1802.04188.pdf.

\bibitem{bil}
\textsc{P. Billingsley}. \textit{Probability and Measure}. John Wiley \& Sons, third edition, 1995. ISBN: 9781118122372.

\bibitem{Radon}
\textsc{Adriaan C. Zaanen}. \textit{Introduction to Operator Theory in Riesz Spaces}. Springer, 1996. ISBN: 3540619895.

\bibitem{vaart}
\textsc{A. W. van der Vaart}. \textit{Asymptotic Statistics}. Cambridge University Press, 1998. ISBN: 9780521784504.

\bibitem{strand}
\textsc{J. L. Strand}. \textit{Random ordinary differential equations}. Journal of Differential Equations 7, 538--553 (1970). doi: 10.1016/0022-0396(70)90100-2.

\bibitem{llibre_powell}
\textsc{Gabriel J. Lord}, \textsc{Catherine E. Powell} and \textsc{Tony Shardlow}. \textit{An Introduction to Computational Stochastic PDEs}. Cambridge Texts in Applied Mathematics, Cambridge University Press, New York, 2014.  

\bibitem{rudin}
\textsc{Walter Rudin}. \textit{Principles of Mathematical Analysis}. International Series in Pure \& Applied Mathematics, third edition, 1976.  

\bibitem{salsa}
\textsc{Sandro Salsa}. \textit{Partial Differential Equations in Action From Modelling to Theory}. Springer, 2009.  
\bibitem{l2}
\textsc{L. Villafuerte}, \textsc{C. A. Braumann}, \textsc{J. C. Cortés} and \textsc{L. Jódar}. \textit{Random differential operational calculus: Theory and applications}. Computers \& Mathematics with Applications 59(1) (2010) 115--125. doi:10.1016/j.camwa.2009.08.061.

\bibitem{apostol}
\textsc{Tom M. Apostol}. \textit{Mathematical Analysis}. Addison-Wesley Publishing Company, Second edition, 1981.  

\bibitem{pitman}
\textsc{Jim Pitman}. \textit{Probability}. Springer, 1993. ISBN: 9780387979748.

\bibitem{llarg}
\textsc{D. Applebaum}, \textsc{B. V. R. Bhat}, \textsc{J. Kustermans} and \textsc{J. M. Lindsay}. \textit{Quantum Independent Increment Processes I, From Classical Probability to Quantum Stochastic Calculus}. Springer, 2005.  

\bibitem{martingales}
\textsc{D. Williams}. \textit{Probability with Martingales}. Cambridge University Press, 1991. ISBN: 0521406056.

\bibitem{truncate}
\textsc{J. F. Lawless}. \textit{Truncated Distributions}. Wiley StatsRef: Statistics Reference Online, 2014. doi: 10.1002/9781118445112.stat04426.

\end{thebibliography}
\end{document}